\documentclass{amsart}

\title[KP solitons under the Gel'fand-Dickey reductions]
{Non-crossing permutations for the KP solitons under the Gel'fand-Dickey reductions and the vertex operators}

\author{Shilong Huang$^a$, Yuji Kodama$^{a,b}$, and Chuanzhong Li$^a$}
\date{\today}
\address{$^A$ College of Mathematics and Systems Science, Shandong University of Science and Technology, Qingdao, 266590, China}
\email{shilonghhqu@163.com, lichuanzhong@sdust.edu.cn}
\address{$^B$ Department of Mathematics, The Ohio State University,
Columbus, OH 43210}
\email{kodama.1@osu.edu}
\subjclass[2000]{}

\usepackage{url}
\usepackage{color}
\usepackage{rotating}

\countdef\x=23
\countdef\y=24
\countdef\z=25
\countdef\t=26

\def\tbox(#1,#2)#3{
\x=#1 \y=#2
\multiply\x by 12
\multiply\y by 12
\z=\x \t=\y
\advance\z by 12
\advance\t by 12
\psline(\x,\y)(\x,\t)(\z,\t)(\z,\y)(\x,\y)
\advance\x by 6
\advance\y by 6
\rput(\x,\y){{\bf #3}}}

  \usepackage{paralist}
  \usepackage{graphics} 
  \usepackage{epsfig} 
\usepackage{graphicx}
\usepackage{epstopdf}
\usepackage{subfigure}
\usepackage{array}

\usepackage{float}
\usepackage{caption}
\usepackage{CJK}

\usepackage{amssymb}
\usepackage{graphicx}
\usepackage{epstopdf}
\usepackage{amscd}
\usepackage{appendix}
\usepackage{graphics}
\usepackage{tikz}
\usepackage{mathdots}
\def\proof{\par{\it Proof}. \ignorespaces}
\def\endproof{{\ \vbox{\hrule\hbox{%
     \vrule height1.3ex\hskip0.8ex\vrule}\hrule }}\par}

\theoremstyle{definition}

\theoremstyle{remark}

\numberwithin{equation}{section}

\let\trueint=\int
\let\truesum=\sum
\def\int{\mathop{\textstyle\trueint}\limits}
\def\sum{\mathop{\textstyle\truesum}\limits}
\let\<=\langle
\let\>=\rangle

\def\sech{\mathop{\rm sech}\nolimits}

\def\Gr{\text{Gr}}

\def\t{\mathbf{t}}
\def\0{\mathbf{0}}

\def\edge{\ar@{-}}
\def\dedge{\ar@{.}}

\usepackage{amsmath, amsthm, amssymb, amsbsy,amsfonts,amscd}
\usepackage[mathscr]{eucal}
\usepackage{epsfig}
\usepackage{young}
\newtheorem{theorem}{Theorem}[section]

\newtheorem{definition}[theorem]{Definition}
\newtheorem{proposition}[theorem]{Proposition}
\newtheorem{lemma}[theorem]{Lemma}

\newtheorem{example}[theorem]{Example}
\newtheorem{corollary}[theorem]{Corollary}

\newtheorem{remark}[theorem]{Remark}

\usepackage{amsfonts}
\usepackage{xy}
\usepackage{amssymb}
\usepackage{amsmath}

\newcommand{\R}{\mathbb R}

\newcommand{\thmrefer}[1]{\renewcommand\thetheorem
  {\protect\ref{#1}}\addtocounter{theorem}{-1}}

\xyoption{all}
\CompileMatrices

\begin{document}

\begin{abstract}

We give a classification of the \emph{regular} soliton solutions of the KP hierarchy, referred to as the \emph{KP solitons}, under the Gel'fand-Dickey $\ell$-reductions in terms of the permutation of the symmetric group. As an example, we show that the regular soliton solutions of the (good) Boussinesq equation as the 3-reduction can have \emph{at most} one resonant soliton in addition to two sets of solitons propagating in opposite directions. We also give a systematic construction of these soliton solutions for the $\ell$-reductions using the vertex operators. In particular, we show that the \emph{non-crossing} permutation gives the regularity condition for the soliton solutions.
\end{abstract}

\maketitle

\tableofcontents

\section{Introduction}

The Kadomtsev-Petviashvili (KP) equation is a two-dimensional nonlinear dispersive wave equation given in the form,
\begin{equation}\label{eq:KP}
 \frac{\partial}{\partial x}\left(-4\frac{\partial u}{\partial t}+6u\frac{\partial u}{\partial x}+\frac{\partial^3 u}{\partial x^3}\right)+3\frac{\partial^2u}{\partial y^2}=0.
\end{equation}
The KP equation is integrable, and has infinitely many symmetries represented by the commuting flows. The parameters of the flows are expressed by $\{t_j:j=1,2,\ldots\}$, in particular, we denote $t_1=x, t_2=y, t_3=t$ for the KP equation. The set of all the flows defines the KP hierarchy, whose first member is the KP equation, and the variable $u$ is now considered as
a function of $\t=(t_1,t_2,\ldots)$. The real regular soliton solutions of the KP equation, referred to as the {\it KP solitons}, has been recently classified in \cite{CK:08, KW:14} (see the book \cite{K:17} for the review of these results).
Each KP soliton can
be determined by a pair of two real data $(\kappa, A)$, referred to as {\it soliton parameters}, where $\kappa=(\kappa_1,\ldots,\kappa_M)\in\mathbb{R}^M$ and an $N\times M$ matrix $A\in\text{Gr}(N,M)_{\ge 0}$, the totally nonnegative (tnn) {Grassmannian} defined by the set of $N$-dimensional subspaces in $\mathbb{R}^M$ with all nonnegative maximal minors of $A$. We also recall that the element $A\in\text{Gr}(N,M)_{\ge0}$ can be uniquely parameterized by the
derangement (permutation without fixed point, i.e. no 1-cycle) of symmetric group $S_M$, denoted by $\pi(A)\in S_M$.
Then the soliton solution $u(x,y,t)$ is given in terms of the $\tau$-function (see e.g. \cite{K:17}),
\begin{equation}\label{eq:u-tau}
u(x,y,t)=2\frac{\partial^2}{\partial x^2}\ln \tau(x,y,t),
\end{equation}
where the $\tau$-function is given by
\begin{equation}\label{eq:tau}
\tau(x,y,t)=\text{det}\left(AE^T(x,y,t)\right).
\end{equation}
The $N\times M$ matrix function $E(x,y,t)$ is defined by
\[
E(x,y,t)=\begin{pmatrix}
E_1 & E_2 &\cdots & E_M\\
\kappa_1 E_1& \kappa_2 E_2 &\cdots & \kappa_M E_M\\
\vdots &\vdots & \ddots & \vdots\\
\kappa_1^{N-1}E_1& \kappa_2^{N-1}E_2&\cdots &\kappa_M^{N-1}E_M
\end{pmatrix}\qquad\text{with}\quad E_j=e^{\kappa_jx+\kappa_j^2y+\kappa_j^3t}.
\]

In this paper, we classify the KP soliton solutions under the Gel'fand-Dickey $\ell$-reductions (hereafter simply $\ell$-reductions) of the KP hierarchy.
For example, the $2$-reduction implies $\partial u/\partial t_2=0~(t_2=y)$, and the KP equation with the boundary condition $u\to 0$ as $x\to \infty$ gives
\[
-4u_t+6uu_x+u_{xxx}=0,
\]
where $u_x=\partial u/\partial x$ etc. The KdV equation admits $N$-soliton solution, whose soliton parameters $(\kappa, A)$ are given by
\begin{align*}
\kappa&=(\kappa_1,\ldots,\kappa_N,-\kappa_N,\ldots,-\kappa_1)\in\mathbb{R}^{2N},\\[1.0ex]
A&=\begin{pmatrix}
1&0&\cdots &\cdots &\cdots&\cdots&0&a_N\\
0&1&\cdots&\cdots&\cdots&\cdots& a_{N-1}&0\\
\vdots &\ddots&\ddots& & &\iddots &\iddots &\vdots\\
0&\cdots&0&1&a_1&0&\cdots&0
\end{pmatrix}\in \text{Gr}(N,2N)_{\ge0},
\end{align*}
where $\kappa_1<\cdots<\kappa_N<0~(\kappa_{j}=-\kappa_{2N-j+1})$ and $a_k=(-1)^{k-1}|a_k|$.
For example, one $(N=1)$ soliton is determined by $\kappa=(\kappa_1,-\kappa_1)$ and $A=(1,a)\in\Gr(1,2)$, and it has the form,
\begin{equation*}\label{eq:KdVsol}
u(x,t)=2\kappa_1^2\sech^2\kappa_1(x+\kappa_1^2t+x_0),
\end{equation*}
where $x_0=\frac{1}{2\kappa_1}\ln a$. The derangement corresponding to the matrix $A$ of the $N$-soliton solution is given by the product of 2-cycles,
\begin{equation}\label{eq:KdVder}
\pi(A)=\prod_{k=1}^N(k,2N-k+1),
\end{equation}
where 2-cycle $(i,j)$ implies the transposition $i\leftrightarrow j$. Note here that all the 2-cycles in \eqref{eq:KdVder} form a {\it nesting} (see below, also \cite{CK:08a} for the definition). Each soliton  with the parameter $\kappa=(\kappa_k,-\kappa_k)$ in the $N$-soliton solution is represented by the transposition $(k,2N-k+1)$, which is considered as the derangement associated to $A_k\in\Gr(1,2)_{\ge0}$, i.e. $\pi(A_k)=(k,2N-k+1)$ and $\pi(A)=\prod_k\pi(A_k)$.

One should also note that the $\kappa$-parameters are the roots of the second order polynomials of $\kappa$,
\begin{equation}\label{eq:KdVpoly}
\Phi_2(\kappa,\alpha):=\kappa^2-\alpha=0, \quad \text{with}\quad \alpha=\kappa_k^2,~\,\,k=1,\ldots,N,
\end{equation}
which is a direct consequence of the 2-reduction of the KP hierarchy. We refer to the polynomial $\Phi_2(\kappa,\alpha)=0$ as the \emph{spectral curve} of the soliton solution for the 2-reduction, and will extend \eqref{eq:KdVpoly} to the general case of $\ell$-reductions.

The (good) Boussinesq equation is given by the 3-reduction
$\partial u/\partial t_3=0$, i.e.
\begin{equation}\label{Boussinesq}
(6uu_x+u_{xxx})_x+3u_{yy}=0.
\end{equation}
In the physical coordinates, the variable $y$ should be considered as the time variable.
One should note that under the boundary condition $u\to 0$ as $x\to\infty$, this equation does not have
a regular soliton solution as the steady propagating wave $u(x,y)=\phi(x+ay)$ with any constant $a$. In order to obtain a regular soliton solution,
one needs to have a nonzero boundary condition $u\to u_0\ne0$ as $x\to\infty$. With the change $u\to u+u_0$, Eq. \eqref{Boussinesq} becomes
\begin{equation}\label{B-shift}
6u_0u_{xx}+(6uu_x+u_{xxx})_x+3u_{yy}=0.
\end{equation}
One can easily check that this equation admits a regular soliton solution when $u_0<0$. Note that \eqref{B-shift} can be also obtained by change of coordinates in the KP equation with $x\to x+ct$, i.e.
\[
(-4u_t-4cu_x+6uu_x+u_{xxx})_x+3u_{yy}=0.
\]
Considering the 3-reduction for new coordinate, i.e. $u_t=0$, and identifying $c=-\frac{3}{2}u_0$, this equation gives \eqref{B-shift}. As we will show, this change of coordinates is a crucial step to classify the regular soliton solutions under the $\ell$-reductions for $\ell\ge 3$. We then choose the $\kappa$-parameters as the roots of the cubic polynomial of $\kappa$ for given constant $\alpha$,
\[
\Phi_3(\kappa,\alpha)=\kappa^3-c\kappa-\alpha=0,
\]
which is the spectral curve of solitons for the 3-reduction.
As one of the main results (Theorem \ref{thm:B}),
we give the regular soliton solutions of the Boussinesq equation in \eqref{B-shift}, {whose soliton matrix $A$} is parametrized by the derangement in the cycle notation either
\begin{align*}
\pi(A)&=(a_1,a_2)\prod_{k=1}^{n_1}(b_k,b_{2n_1-k+1})\prod_{l=1}^{n_2}(c_l,c_{2n_2-l+1})
\qquad\text{or}\\
\pi(A)&=(a_1,a_2,a_3)\prod_{k=1}^{n_1}(b_k,b_{2n_1-k+1})\prod_{l=1}^{n_2}(c_l,c_{2n_2-l+1}),
\end{align*}
where the cycles are all {\it non-crossing}, e.g. there is no case like $a_1<b_i<a_2<b_j$ or $b_i<c_j<b_k<c_l$ etc (see Definition \ref{def:NC} below for the details).
One should note that the regular soliton solution can include \emph{at most} one 3-cycle, which represents a resonant soliton solution with Y-shape, sometimes called Y-soliton
(see e.g. \cite{K:17}). This result implies that there are two sets of line solitons showing in 2-cycles propagating opposite directions, and each soliton gets positive (repulsive) phase shift by interacting with solitons from the same set, and negative (attractive) phase shift from the other set. Thus, the set of those solitons may provide a bidirectional model of
soliton gas as recently discussed in \cite{BD:24}.

\begin{remark}
As far as we know, there was no classification of ``regular'' soliton solutions of the (good) Boussinesq equation (see e.g. \cite{RS:17}). Our result is then gives
the first classification of regular soliton solutions of the Boussinesq equation, which states that there is no regular soliton solutions including more than one resonant solution.
Also note that singular solutions are easily obtained.
\end{remark}

We extend these results to the general $\ell$-reduction as follows. First define a monic polynomial of degree $\ell$
of $\kappa$, referred to as a \emph{spectral curve of $\ell$-reduction}, given by
\begin{equation}\label{eq:SC}
\Phi_\ell(\kappa,\alpha):=\varphi_\ell(\kappa)-\alpha=\prod_{j=1}^\ell(\kappa-\kappa_j[\alpha])=0.
\end{equation}
Here the polynomial $\varphi_\ell(\kappa)$ can be also obtained by the change of coordinates in the KP hierarchy
as in the case of 3-reduction, $\varphi_3(\kappa)=\kappa^3-c\kappa$. The main purpose of the spectral curve in the form \eqref{eq:SC} is to find the real roots for the soliton parameter $\kappa\in\R^M$
i.e. $\kappa_j[\alpha]\in\R$ and $\kappa_j[\alpha]\ne\kappa_l[\alpha]$ if $j\ne l$.

Then our main theorem (Theorem \ref{thm:main}) can be stated as follows.

\bigskip
\noindent
\emph{Theorem} \ref{thm:main}
Define the soliton parameter $\kappa\in\R^M$ as the ordered set of
\[
\mathcal{K}_\ell:=\bigcup_{i=1}^K\{\kappa_{j}[\alpha_i]:j\in\mathcal{I}[\alpha_i]\},
\]
where each $\kappa_{j}[\alpha_i]$ is a root of $\Phi_\ell(\kappa,\alpha_i)=0$, and $\mathcal{I}[\alpha_i]$ is a subset of $[\ell]:=\{1,\ldots,\ell\}$ with $|\mathcal{I}[\alpha_i]|=m_i$. Here $M=m_1+\cdots+m_K$. For each $m_i$, take $1\le n_i\le m_i-1$, and set $N=n_1+\cdots+n_K$.
Then for each pair $(n_i,m_i)$, construct the matrix element $A[\alpha_i]\in\text{Gr}(n_i,m_i)_{\ge0}$, so that the corresponding derangements $\pi(A[\alpha_i])$ are mutually \emph{non-crossing} for $i=1,\ldots,K$. Then the regular soliton solution of the $\ell$-reduction is given by the soliton parameters $\kappa\in\R^M$ and $A[\alpha_1,\ldots,\alpha_K]\in\Gr(N,M)_{\ge0}$ associated with a ($\kappa$-)direct sum of the matrices $A[\alpha_i]$,
\[
{\widehat\bigoplus}_{i=1}^K  A[\alpha_i],
\]
where $\widehat\oplus$ implies that the columns in the matrices are ordered according to the ordered set $\mathcal{K}_\ell$ of the $\kappa$-parameters.
Also the derangement $\pi(A[\alpha_1,\ldots,\alpha_K])$ is given by
\[
\pi(A[\alpha_1,\ldots,\alpha_K])=\prod_{i=1}^K\pi(A[\alpha_i]),
\]
that is, the soliton solution generated by the $\tau$-function with $(\kappa,A[\alpha_1,\ldots,\alpha_K])$ consists of these soliton solutions associated with $(\kappa[\alpha_i],A[\alpha_i])$, where $\kappa[\alpha_i]=\{\kappa_j[\alpha_i]:j\in\mathcal{I}[\alpha_i]\}$.

\medskip
\noindent
Here what we mean by the non-crossing in the derangements $\pi(A[\alpha_i])$ is that any cycle in $\pi(A[\alpha_i])$ has no crossing with all the cycles in other derangemants $\pi(A[\alpha_j])$ $(i\ne j)$ (Definition \ref{def:NC}).

\bigskip

 The paper is organized as follow.
 In Section \ref{sec:background}, we provide the background information on the KP hierarchy and the Gel'fand-Dickey $\ell$-reductions. We also give the determinant formula of the KP solitons in terms of the $\tau$-functions (see also e.g. \cite{K:17}). In Section \ref{sec:SC}, we define a modified form of the $\ell$-reduction for real regular solitons of the $\ell$-reduction. This is the main section of the paper.
We introduce an $\ell$-th polynomial $\Phi_\ell(\kappa,\alpha)=\varphi_\ell(\kappa)-\alpha=0$, referred to as the spectral curve of the $\ell$-reduction, whose roots give real soliton $\kappa$-parameters. We show that the polynomial $\varphi_\ell(\kappa)$ is a versal deformation of the degenerate polynomial $\kappa^\ell$, and the deformation parameters are obtained by the change of coordinates. We then consider a matrix $A[\alpha]\in \Gr(n,m)_{\ge 0}$ together with the roots of the spectral curve for a soliton solution.
A general solution is then constructed by several parameters $\alpha_j$, which give the roots $\kappa[\alpha_j]\in\R^{m_j}$ and the matrix $A[\alpha_j]\in\Gr(n_j,m_j)_{\ge 0}$.
We then define a $\kappa$-direct sum of the matrices (Definition \ref{def:SumA}), and introduce the notion of \emph{non-crossing} matrices (Definition \ref{def:NC}).
Then we show that the matrices $A[\alpha_j]\in\Gr(n_j,m_j)_{\ge0}$ are mutually non-crossing, then the the direct sum of these matrices can be totally nonnegative, that is,
the soliton generated by the $\kappa$-direct sum is regular (Theorem \ref{thm:main}).
In Section \ref{sec:Boussinesq}, we give the detailed study of the Boussinesq case, the 3-reduction.  The main result is explained above (Theorem \ref{thm:B}).
We also discuss a possible application of the result to a bidirectional model of soliton gas (see e.g. \cite{BD:24}).
In the final section, Section \ref{sec:VO}, we give a systematic method to construct the KP solitons under the $\ell$-reductions using the vertex operators.
We show that the non-crossing property of the matrices associated with vertex operator gives the regularity of the solitons generated by applying several vertex operators
(Proposition \ref{prop:regularity}).
We also remark that the regularity of the soliton solutions generated by the vertex operators has not been discussed before.

\medskip

\section{Background}\label{sec:background}
Here we give a brief review of the KP hierarchy and
the Gel'fand-Dickey $\ell$-reductions of the KP hierarchy. We also give the Wronskian
formula of the $\tau$-function for the KP solitons, in which we provide the notions of the \emph{totally nonnegative} Grassmannians and the permutations to label their elements.
Most of the materials in this section can be also found in \cite{K:17}.

\subsection{The KP hierarchy}
The KP hierarchy is formulated on the basis with a pseudo-differential operator,
\begin{equation}\label{eq:LaxOP}
L=\partial + u_2\partial^{-1}+u_3\partial^{-2}+\cdots,
\end{equation}
where $\partial$ is a derivative satisfying $\partial\partial^{-1}=\partial^{-1}\partial=1$ and the generalized Leibniz rule for a smooth function $f$,
\[
\partial^{\nu}f=\sum_{k=0}^\infty\binom{\nu}{k}(\partial_1^kf)\partial^{\nu-k},\qquad \nu\in\mathbb{Z}.
\]
Note that the series terminates if and only if $\nu$ is a nonnegative integer.
The functions $u_k$'s in $L$ depend on an infinite number of variables $\t=(t_1,t_2,\ldots)$. Each variable $t_n$ in $L$ gives a parameter of the flow in the hierarchy, which is defined in the Lax form,
\begin{equation}\label{eq:Lax}
\frac{\partial L}{\partial t_n}=[B_n,\,L]\qquad\text{with}\qquad B_n=(L^n)_{\ge 0}\quad (n=1,2,\ldots),
\end{equation}
where $(L^n)_{\ge 0}$ represents the polynomial (differential) part of $L^n$ in $\partial$.
The solution of the KP equation \eqref{eq:KP} is given by $u=2u_2$. The compatibility among the equations in \eqref{eq:Lax}, gives
\begin{equation*}\label{eq:ZS}
\frac{\partial B_n}{\partial t_m}-\frac{\partial B_m}{\partial t_n}+[B_n,\,B_m]=0,
\end{equation*}
which is called the Zakharov-Shabat equations.

The Lax operator \eqref{eq:LaxOP} can be expressed in the form,
\[
L=W\partial W^{-1},
\]
where $W$ is called the dressing operator given in the form,
\[
W=1-w_1\partial^{-1}-w_2\partial^{-2}-\cdots.
\]
Then the functions $u_i$'s in $L$ can be expressed by $w_j$'s in $W$.
For example, we have
\[
u_2=\frac{\partial w_1}{\partial x},\qquad u_3=\frac{\partial w_2}{\partial x}+w_1\frac{\partial w_1}{\partial x},\quad\ldots.
\]
Then, from the Lax equation, the dressing operator $W$ satisfies
\begin{equation}\label{eq:Sato}
\frac{\partial W}{\partial t_n}=B_nW-W\partial^n\qquad \text{for}\quad n=1,2,\cdots,
\end{equation}
which is sometimes called the Sato equation. The KP soliton can be obtained as a special solution of the Sato equation as explained below.

\subsection{$N$ truncation and the $\tau$-function}
Here we explain the Wronskian formula of the $\tau$-function, which is obtained by truncating $W$.
A finite truncation of $W$ with some positive integer $N$ is given by
\[
W=1-w_1\partial^{-1}-w_2\partial^{-2}-\cdots-w_N\partial^{-N}.
\]
The invariance of the truncation under \eqref{eq:Sato} leads to the $N$-th order differential equation for some function $f=f(\t)$,
\[
W\partial^N f=f^{(N)}-w_1f^{(N-1)}-w_2f^{(N-2)}-\cdots-w_Nf=0,
\]
where $f^{(n)}=\frac{\partial^n}{\partial x^n} f$. Let $\{f_i:i=1,\ldots,N\}$ be a fundamental set of solutions of
the equation $W\partial^Nf=0$. Then the coefficients $w_i(\t)$'s are given by
\begin{equation}\label{eq:wi}
w_i(\t)=-\frac{1}{\tau(\t)}p_i(-\tilde\partial)\tau(\t)\qquad\text{for}\quad i=1,\ldots,N,
\end{equation}
where $p_i({\bf x})$ is the elementary Schur polynomial of degree $i$ and $\tilde\partial=(\partial_1,\frac{1}{2}\partial_2,\frac{1}{3}\partial_3,\ldots)$, which is defined by
\begin{equation}\label{eq:xi}
e^{\xi(\t,k)}=\sum_{n=0}^\infty p_n(\t)k^n,\qquad \xi(\t,k):=\sum_{n=1}^\infty k^nt_n.
\end{equation}
And $\tau(\t)$ is called the $\tau$-function, which is expressed by the Wronskian form (Cramer's rule),
\begin{equation}\label{eq:WrA}
\tau(\t)=\text{Wr}(f_1,f_2,\ldots,f_N)=\left|\begin{matrix}
f_1 & f_2 & \cdots & f_N \\
f_1^{(1)}&f_2^{(1)}&\cdots&f_N^{(1)}\\
\vdots &\vdots &\ddots &\vdots\\
f_1^{(N-1)}&f_2^{(N-1)}&\cdots&f_N^{(N-1)}
\end{matrix}\right|.
\end{equation}
For the time-evolution of the functions $f_i(\t)$, we consider the following (heat) hierarchy,
\begin{equation}\label{eq:Df}
\frac{\partial f_i}{\partial t_n}=\frac{\partial^n f_i}{\partial x^n}\qquad \text{for}\quad 1\le i\le N,\quad  n=1,2,\ldots,
\end{equation}
which gives the solution of the Sato equation \eqref{eq:Sato}. Then the solution of the KP hierarchy
can be expressed in terms of the $\tau$-function by \eqref{eq:u-tau}, i.e.
\begin{equation*}
u(\t)=2u_2(\t)=2\frac{\partial w_1(\t)}{\partial x}=2\frac{\partial^2}{\partial x^2}\ln\tau(\t).
\end{equation*}


\subsection{KP solitons}
The soliton solutions are defined by a finite set of exponential solutions of \eqref{eq:Df}. Let $\{f_i(\t):i=1,\ldots,N\}$  be the set of solutions given by
\begin{equation*}\label{eq:KPf}
f_i(\t)=\sum_{j=1}^Ma_{i,j}E_j(\t),\qquad\text{with}\qquad E_j(\t)=e^{\xi_j(\t)} \quad\text{and}\quad \xi_j(\t)=\xi(\t,\kappa_j),
\end{equation*}
where $\kappa\in\R^M$ whose elements are ordered as $(\kappa_1<\cdots<\kappa_M)$, and $A=(a_{i,j})$ is an $N\times M$ real matrix with full rank.
Let $I=\{i_1<i_2<\cdots<i_N\}$ be an $N$ ordered subset of the index set $[M]:=\{1,2,\ldots,M\}$.
Then the $\tau$-function \eqref{eq:WrA} becomes \eqref{eq:tau}, and using the Binet-Cauchy lemma, the $\tau$-function can be expressed by
\begin{equation}\label{eq:tau-Exp}
\tau(\t)=\sum_{I\in\mathcal{M}(A)}\Delta_I(A)E_I(\t),\quad\text{with}\quad E_I(\t)=\prod_{k>l}(\kappa_{i_k}-\kappa_{i_l})E_{i_1}\cdots E_{i_N},
\end{equation}
where $\Delta_I(A)$ is the $N\times N$ minor of $A$ associated with the ordered subset $I$, and $\mathcal{M}(A)$ is the matroid defined by
\begin{equation*}\label{eq:matroid}
\mathcal{M}(A):=\left\{I\in\binom{[M]}{N}:\Delta_I(A)\ne 0\right\}.
\end{equation*}
Note here that the ordering in the parameters $\kappa=(\kappa_1,\ldots,\kappa_M)$ implies $E_I(\t)>0$ for all $I\in\binom{[M]}{N}$ and $\t\in\R^M$.
Then it was shown in \cite{KW:13} that the regular soliton solutions, referred to as \emph{KP solitons}, are obtained if and only if
$A=(a_{i,j})\in\text{Gr}(N,M)_{\ge0}$ is an element of the totally nonnegative Grassmannian, which is defined as
\[
\Gr(N,M)_{\ge0}:=\left\{A\in \Gr(N,M):\Delta_I(A)>0~\text{for all}~I\in\mathcal{M}(A)\right\}.
\]
For example, one line-soliton solution is determined by $\kappa=(\kappa_i,\kappa_j)$ and $A=(1,a)\in\text{Gr}(1,2)_{\ge0}$ for $a>0$, i.e. $\tau(\t)=E_i(\t)+aE_j(\t)$ and
\begin{equation*}\label{eq:1sol}
u(\t)=2\frac{\partial^2}{\partial t_1^2}\ln\tau(\t)=\frac{(\kappa_i-\kappa_j)^2}{2}\sech^2\frac{1}{2}(\xi_i(\t)-\xi_j(\t)+\ln a),
\end{equation*}
which is referred to as a line-soliton solution of $[i,j]$ type, or simply $[i,j]$-soliton.

Then we have the following theorem \cite{CK:08, CK:09}.
\begin{theorem}
Let $\{i_1,\ldots,i_N\}$ be the pivot set and $\{j_1,\ldots,j_{M-N}\}$ be the nonpivot set of $A\in\text{Gr}(N,M)_{\ge0}$.
Then there exists a unique derangement $\pi$ of the symmetric group $S_M$ associated with the matrix $A$, denoted by $\pi(A)\in S_M$, so that
 the KP soliton has the following asymptotic structure.
\begin{itemize}
\item[(a)] For $y\gg 0$, there are $[i_n,\pi(i_n)]$-solitons with $\pi(i_n)>i_n$ for $n=1,\ldots,N$.
\item[(b)] For $y\ll 0$, there are $[\pi(j_m),j_m]$-solitons with $\pi(j_m)<j_m$ for $m=1,\ldots,M-N$.
\end{itemize}
\end{theorem}

It would be useful to express the derangements in the chord diagrams defined as follows (see e.g. \cite{K:17}).
\begin{definition}\label{def:chords}
Consider a line segment with $M$ marked points labeled by the $\kappa$-parameter $(\kappa_1,\ldots,\kappa_M)$.
Then the chord diagram associated with a derangement $\pi$ in the symmetric group $S_M$ is defined by
\begin{itemize}
\item[(a)] if $i<\pi(i)$ (exceedance), then draw a chord joining $\kappa_i$ and $\kappa_{\pi(i)}$ on the upper part of the line, and
\item[(b)] if $j>\pi(j)$ (deficiency), then draw a chord joining $\kappa_j$ and $\kappa_{\pi(j)}$ on the lower part of the line.
\end{itemize}
Note that if the derangement is given by a single $k$-cycle, $\pi=(j_1,\ldots,j_k)$, then all the points $\{\kappa_{j_i};~i=1,\ldots,k\}$ are joined in the chord diagram.
\end{definition}

\begin{example}\label{ex:48}
Let $\pi=(1,3,2,7,5)(4,6,8)\in S_8$ be a permutation in the cycle notation, whose chord diagram is given by
\begin{figure}[H]
  \centering
  \includegraphics[height=2.3cm,width=7.5cm]{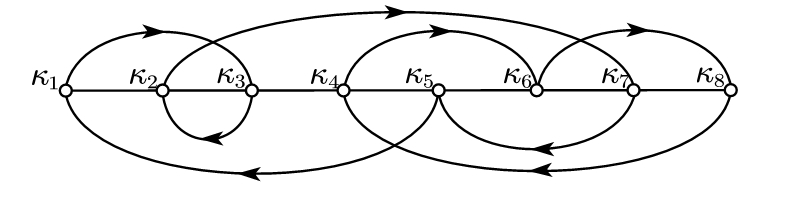}
\end{figure}
\noindent
Then following the method given in \cite{KW:14, K:17}, we have $A\in\text{Gr}(4,8)_{\ge0}$,
\[
A=\begin{pmatrix}
1& & & & a_3a_6a_8& & -a_1a_3a_4a_6a_8 & -b_1\\
 &1 &a_5 & &-a_3a_5a_6&&a_1a_3a_4a_5a_6&b_2 \\
  & & & 1& a_3 & &-a_1a_3a_4&-a_1a_2a_3a_4\\
   & & & &  &1&a_1&a_1a_2
 \end{pmatrix}
 \]
 with $b_1=a_2a_6a_8(a_7+a_1a_3a_4), ~ b_2=a_2a_5a_6(a_7+a_1a_3a_4)$,
 where the blank elements are zero, and all $a_i>0$ for $i=1,\ldots,8$.
 \end{example}

 It was also shown in \cite{C:07, W:05} that the number of free parameters in $A\in\text{Gr}(N,M)_{\ge0}$
 can be found from the chord diagram, and it is given by
 \begin{equation}\label{eq:dimA}
 N+\{\text{\# of crossings}\}+\{\text{\# of cusps in the lower part}\},
 \end{equation}
 which is the dimension of the positroid cell parametrized by the derangement $\pi(A)$ \cite{W:05, KW:14}. The positroid cell decomposition of the \emph{irreducible} totally nonnegative Grassmannian is given by
\[
\Gr(N,M)_{\ge0}^{\text{irr}}=\bigsqcup_{\pi\in D_M}X_\pi,
\]
where $D_M\subset S_M$ is the set of derangements, and the \emph{irreducibility} implies that for each element $A\in\Gr(N,M)_{\ge0}$,
the matrix $A$ satisfies
\begin{itemize}
\item[(a)] there is no zero column,
\item[(b)] there exists at least one nonzero elements in each row besides the pivot.
\end{itemize}
In the example \ref{ex:48}, the dimension $X_\pi$ is given by
\[
\text{dim}(X_\pi)=4+3+1=8,
\]
which is the total number of positive parameters $a_1,\ldots,a_8$.

\subsection{The Gel'fand-Dickey $\ell$-reductions}
Then the Gel'fand-Dickey $\ell$-reduction (sometime called the $\ell$-th generalized KdV hierarchy, see e.g. \cite{KX:21}) is defined by
\[
L^\ell=B_\ell:=(L^\ell)_{\ge 0},
\]
that is, the $\ell$-th power of $L$ becomes a differential operator. This means that the functions $u_i$'s
are determined by $\ell-1$ variables in $L^\ell$ in the form,
\begin{equation}\label{eq:LL}
L^{\ell}=\partial^{\ell}+r_{2}\partial^{\ell-2}+r_{3}\partial^{\ell-3}+\cdots+r_{\ell},
\end{equation}
where the functions $r_{i}$ is given by the differential polynomial of $\{u_{2},\cdots, u_{i}\}$ and their derivatives with respect to $t_1$.
For example, when $\ell=4$, we have
\begin{align*}
r_{2}=4u_{2},\quad\quad\quad r_{3}=6u_{2,t_{1}}+4u_{3},\quad\quad\quad r_{4}=4u_{2, t_{1}t_{1}}+6u_{2}^{2}+6u_{3, t_{1}}+4u_{4}.
\end{align*}

We can also see from $L^\ell=(W\partial^\ell W^{-1})_{\ge 0}$
that these functions $r_n$ have the following form,
\begin{equation}\label{eq:rn}
r_{n+1}=\frac{\ell}{n}\partial_{t_{1}}\partial_{t_{n}}\ln\tau+R_n(r_{2}, \cdots, r_{n})\quad\quad\text{with}\quad\quad 1\leq n \leq \ell-1,
\end{equation}
where $R_n(r_2,\ldots,r_n)$ is the differential polynomial of $\{r_2,\ldots,r_n\}$ with $R_1=0$.
For example, when $\ell=4$, we have
\begin{align*}
r_{2}=4\partial_{t_{1}}^{2}\ln\tau,\quad\quad r_{3}=2\partial_{t_{1}}\partial_{t_{2}}\ln\tau+\partial_{t_{1}}r_{2},\quad\quad r_{4}=\frac{4}{3}\partial_{t_{1}}\partial_{t_{3}}\ln\tau+\frac{1}{2}\partial_{t_{1}}r_{3}-\frac{1}{12}\partial_{t_{1}}^{2}r_{2}+\frac{1}{8}r_{2}^{2}.
\end{align*}
From \eqref{eq:Lax} with \eqref{eq:LL}, the $\ell$-reduction gives the constraints,
\[
\frac{\partial L}{\partial t_{n\ell}}=0\qquad\text{for}\quad n=1,2,\ldots,
\]
that is, all the variables $\{r_i:1\le i\le \ell-1\}$ do not depend on the times $t_{n\ell}$. Following \cite{KX:21}, we define the $\ell$-reduction as follows. Let $\{n_1,\ldots,n_K\}$ be a partition of $N$ with $1\le n_k<\ell$, i.e. $N=n_1+\cdots+n_K$, and define $N_k=n_1+\cdots+n_k$ $(\text{i.e.}~N=N_K)$.
Then the $\ell$-reduction is defined by a condition for the functions $\{f_i:1\le i\le N\}$ in \eqref{eq:Df} given by
\begin{equation}\label{eq:Reduction}
\frac{\partial f_{i}}{\partial t_\ell}=\alpha_k f_{i},\qquad \text{for}\quad N_{k-1}+1\le i\le N_k,\quad 1\le k\le K,
\end{equation}
for some constant $\alpha_k\in\R$ and $N_0=0$. Using \eqref{eq:Df}, this leads to
\[
\frac{\partial f_{i}}{\partial t_\ell}=\frac{\partial^\ell f_{i}}{\partial t_1^\ell}=\alpha_k f_{i}.
\]
Looking for an exponential solution to this equation, we have the degenerate polynomial of $\kappa$, i.e.
\begin{equation}\label{eq:Dpoly}
\kappa^\ell=\alpha_k,
\end{equation}
which has the complex roots, the $\ell$-th roots of $\alpha_k$ (see \cite{KX:21}, where the complex solitons associated with thses roots were discussed).

One should note that the reduction equation \eqref{eq:Reduction} implies that
we have {$f_{i}(\t)=e^{\sum_n\alpha_k^n t_{n\ell}}g_{i}(\t')$} with $\t'=(t_j:j\ne 0~(\text{mod}~\ell))$, which shows that the $\tau$-function is given by $\tau(\t)=\text{Wr}(f_1,\ldots f_N)(\t)=e^{\sum_k\sum_n\alpha_k^nt_{n\ell}}\text{Wr}(g_1,\ldots,g_N)(\t')$. Then the variables $r_n(\t)$ in \eqref{eq:rn} has no dependency on the flow parameters $t_{n\ell}$.

In the next section, we define a versal deformation of the degenerate polynomial \eqref{eq:Dpoly} to find the real and regular solutions.

\section{Spectral curves and the $\tau$-functions for the KP solitons under the $\ell$-reductions}\label{sec:SC}
Here we first introduce the spectral curve defined as a \emph{versal} deformation of the degenerated polynomial \eqref{eq:Dpoly}
in order to obtain a set of exponential functions for the basis of real solitons. Then we construct the $\tau$-function for a particular basis of exponential functions and describe the corresponding soliton solution in terms of the permutation.

\subsection{The spectral curve and the $\tau$-function}
We consider a \emph{versal} deformation of the $\ell$-th degree polynomial \eqref{eq:Dpoly} given by
\begin{equation}\label{eq:Vpoly}
\Phi_\ell(\kappa,\alpha):=\varphi_\ell(\kappa)-\alpha=0\qquad\text{with}\qquad \varphi_\ell(\kappa):=\kappa^\ell-\sum_{j=1}^{\ell-1}c_j\kappa^j,
\end{equation}
where $c_j$ are real constants (deformation parameters).  We particularly choose $c_j$'s , so that \eqref{eq:Vpoly} has $\ell$ distinct real roots $\kappa_j[\alpha]$, which is \eqref{eq:SC}, i.e. the versality implies that
\[
\Phi_\ell(\kappa,\alpha)=\prod_{j=1}^\ell(\kappa-\kappa_j[\alpha])=0.
\]
We call this \emph{spectral curve} for the KP soliton under the $\ell$-reduction.

We note that the versal deformation \eqref{eq:Vpoly} can be realized by the change of coordinates,
{\begin{equation}\label{eq:change}
t_n~\rightarrow~\left\{
\begin{array}{lll}
t_n+c_mt_{k\ell},&\quad\text{if}~n=m+(k-1)\ell\quad\text{with}\quad 1\le m\le \ell-1\quad\text{and}\quad k\geq1,\\
t_{k\ell},&\quad \text{if}~n=k\ell\quad\text{with}\quad k\geq1.
\end{array}\right.
\end{equation}}
In the new coordinates, the derivatives are
{\begin{equation*}\label{eq:Dchange}
\frac{\partial}{\partial t_n}~\rightarrow~\left\{
\begin{array}{lll}
\frac{\partial}{\partial t_n},&\quad\text{if}~n\neq k\ell\quad\text{with}\quad k\geq1,\\
\frac{\partial}{\partial t_{k\ell}}+\sum_{m=1}^{\ell-1}c_m\frac{\partial}{\partial t_{(k-1)\ell+m}},&\quad \text{if}~n=k\ell\quad\text{with}\quad k\geq1.
\end{array}\right.
\end{equation*}}
Then we impose that the condition for the $\ell$-reduction in the new coordinates becomes
\[
\frac{\partial^\ell f_i}{\partial t_1^\ell}-\sum_{n=1}^{\ell-1}c_n\frac{\partial^n f_i}{\partial t_1^n}=\alpha f_i.
\]
As a particular solution of this equation, we consider an exponential function $f(\t)=\exp(\kappa t_1)$. Then the parameter $\kappa$ satisfies $\Phi_\ell(\kappa,\alpha)=0$, hence we have $\ell$ independent exponential solutions,
\[
E_j(\t,\alpha):=\exp\left(\sum_{n=1}^\ell \kappa_j[\alpha]^nt_n\right),\qquad\text{for}\quad j=1,\ldots,\ell.
\]
Notice that in the original coordinates in \eqref{eq:change}, we have
\begin{equation}\label{eq:E}
E_j(\t,\alpha)=e^{\alpha t_\ell}\exp\left(\sum_{n=1}^{\ell-1} \kappa_j[\alpha]^nt_n\right)=:e^{\alpha t_\ell}E_j(\hat{\t},\alpha),
\end{equation}
i.e. $E_j(\t,\alpha)$ is a solution of the $\ell$-reduction \eqref{eq:Reduction}, and $\hat{\t}=(t_1,\ldots,t_{\ell-1})$.

Then the $\tau$-function in \eqref{eq:tau-Exp} is given by
\begin{equation}\label{eq:tau-alpha}
\tau(\hat\t,\alpha)=\sum_{I\in\mathcal{M}(A[\alpha])}\Delta_I(A[\alpha])E_I(\hat\t,\alpha),
\end{equation}
where
\[
E_I(\hat\t,\alpha)=\text{Wr}(E_{i_1},\ldots,E_{i_N})=\prod_{j>k}\left(\kappa_{i_j}[\alpha]-\kappa_{i_k}[\alpha])\right)\cdot \prod_{j=1}^NE_{i_j}(\hat\t,\alpha).
\]

\subsection{Non-crossing permutations and the main theorem}
Here we consider several values of $\alpha\in\R$, say $\alpha_k$ for $k=1,\ldots,K$.
For each $\alpha=\alpha_k$, let $\mathcal{I}[\alpha_k]$ be a subset of $[\ell]=\{1,\ldots,\ell\}$ with {$2\leq|\mathcal{I}[\alpha_k]|=m_k\le \ell$}. We then consider $\Phi_\ell(\kappa,\alpha_k)=0$
in \eqref{eq:Vpoly} for $k=1,\ldots,K$, and define the soliton parameter $\kappa\in\R^M$,
which is the \emph{ordered} set of the roots of $\Phi_\ell(\kappa,\alpha_k)=0$,
\[
\kappa:=\text{ord}\left(\bigcup_{k=1}^K\left\{\kappa_j[\alpha_k]: j\in\mathcal{I}[\alpha_k]\,\right\}\right)=:(\kappa_1<\kappa_2<\cdots<\kappa_M),
\]
where $M=m_1+\cdots+m_K$. Since all $\kappa_j[\alpha_k]$ are distinct, this gives that for each $m\in[M]:=\{1,\ldots,M\}$, there exists a unique pair $(j,k)$ so that $\kappa_m=\kappa_j[\alpha_k]$, which gives a bijection $m=\phi(j,k)$ for each $j\in\mathcal{I}[\alpha_k]$. We also write
\begin{equation}\label{eq:IJ}
\hat{\mathcal{I}}[\alpha_k]:=\phi(\mathcal{I}[\alpha_k]):=\{\, m=\phi(j,k): j\in\mathcal{I}[\alpha_k]\,\}.
\end{equation}
With this ordering, we define
\begin{equation*}\label{eq:f}
f_i(\t,\alpha_k)=\sum_{j\in{\hat{\mathcal{I}}}[\alpha_k]}a_{i,j}(\alpha_k)E_j(\t),\quad \text{for}\quad N_{k-1}+1\le i\le N_k,
\end{equation*}
where $N_k=n_1+\cdots+n_k$ with $1\le n_k\le m_k-1$. Note that $E_j(\t)=\exp(\sum_{n=1}^\infty\kappa_j^nt_n)$.
Here we also take $A[\alpha_k]:=(a_{i,j}(\alpha_k))\in\Gr(n_k,m_k)_{\ge0}$, and let $\pi(A[\alpha_k])$ be the corresponding permutation in the cycle notation,
\[
\pi(A[\alpha_k])=\prod_{p=1}^{P_k}(j^{(p)}_1,\ldots,j^{(p)}_{k_p}),
\]
where $P_k$ is the number of cycles, and $(j^{(p)}_1,\ldots,j^{(p)}_{k_p})$ is a $k_p$-cycle for $j^{(p)}_i\in\hat{\mathcal{I}}[\alpha_k]$.

Let $N=N_K=n_1+\cdots+n_K$. Then we define an $N\times M$ matrix combining all the matrices $A[\alpha_k]$ for $k=1,\ldots,K$, which we call a $\kappa$-direct sum.
\begin{definition}\label{def:SumA}
A \emph{$\kappa$-direct sum} of the matrices $A[\alpha_k]$ for $k=1,\ldots,K$ is defined by
\[
\widehat{\bigoplus}_{k=1}^KA[\alpha_k]:=\left(a_{i,m}\right)_{1\le i\le N,~ 1\le m\le M},
\]
where $m=\phi(j,k)$ with $j\in\mathcal{I}[\alpha_k]$ for $1\le k\le K$. In particular, the row index $i$ is assigned so that the direct sum is an element of $\Gr(N,M)$,
 i.e.
 it is in the reduced row echelon form (RREF). Notice that in general, it is not totally nonnegative.
 \end{definition}

\begin{example}\label{ex:17}
Consider a case with $\ell=7$, and two different $\alpha$'s in the order $\alpha_1<\alpha_2$. Take
\[
\mathcal{I}[\alpha_1]=\{1,4,5,7\},\qquad \mathcal{I}[\alpha_2]=\{1,2,4,5,6\}.
\]
The spectral curve can be expressed as in Figure \ref{fig:SC7}.
\begin{figure}[H]
  \begin{minipage}[t]{1\linewidth}
  \centering
  \includegraphics[height=3cm,width=11cm]{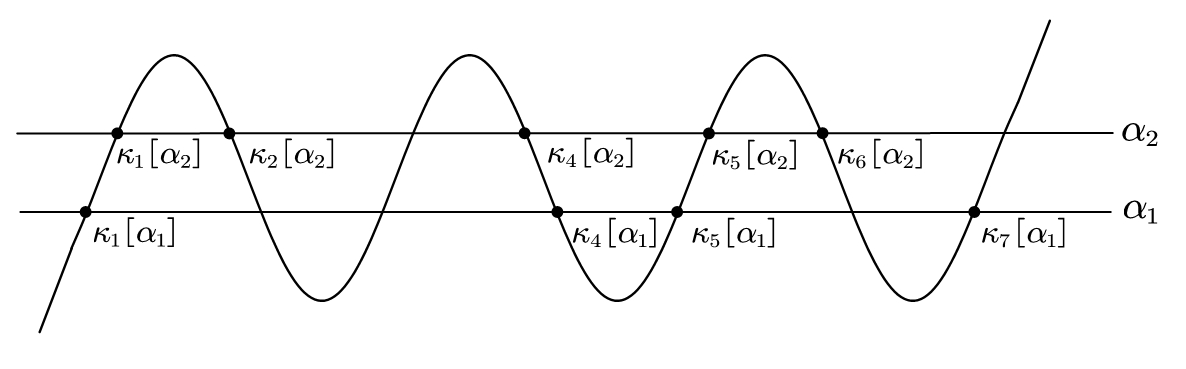}
  \end{minipage}
  \caption{The spectral curve $\Phi_7(\kappa,\alpha)=0$ with the roots $\kappa_j[\alpha_i]$ for $i=1,2$.\label{fig:SC7}}
\end{figure}
\noindent These roots satisfy the following order, giving the bijection $\phi:\kappa_{j}[\alpha_k]\to \kappa_m$,
\[
\kappa_1=\kappa_1[\alpha_1]<\kappa_1[\alpha_2]<\kappa_2[\alpha_2]<\kappa_4[\alpha_2]<\kappa_4[\alpha_1]<\kappa_5[\alpha_1]<\kappa_5[\alpha_2]<\kappa_6[\alpha_2]<\kappa_7[\alpha_1]=\kappa_9,
\]
which gives
\[
\hat{\mathcal{I}}[\alpha_1]=\{1,5,6,9\},\qquad \hat{\mathcal{I}}[\alpha_2]=\{2,3,4,7,8\}.
\]
As an example, consider the following permutations,
\begin{align*}
\pi(A[\alpha_1])=(1,6,9,5),\qquad \pi(A[\alpha_2])=(2,3,4)(7,8),
\end{align*}
whose chord diagram is
\begin{figure}[H]
  \begin{minipage}[t]{1\linewidth}
  \centering
  \includegraphics[height=2.2cm,width=7cm]{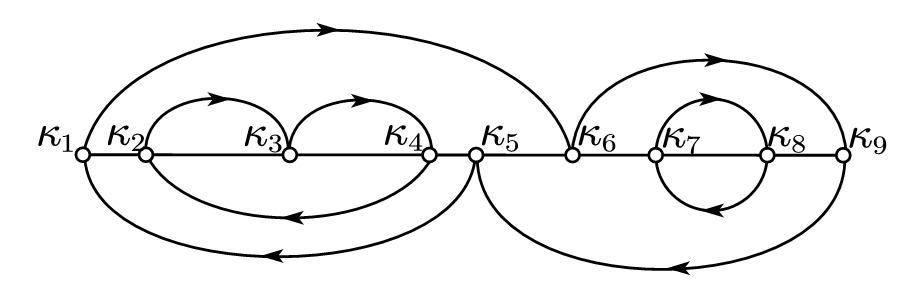}
  \end{minipage}
\end{figure}
\noindent The corresponding matrices $A[\alpha_1]\in\text{Gr}(2,4)_{\ge0}$ and $A[\alpha_2]\in\text{Gr}(3,5)_{\ge 0}$ are given by
\[
A[\alpha_1]=\begin{pmatrix}
1&a& &-b\\
 & &1&c
 \end{pmatrix},\qquad\text{and}\qquad
 A[\alpha_2]=\begin{pmatrix}
 1& &-d& & \\ &1 &e& & \\ & & &1 &f
 \end{pmatrix},
 \]
 where $a, b,\ldots,f$ are positive constants,
 (see e.g. \cite{K:17} for the method to construct $A\in\text{Gr}(N,M)_{\ge0}$ from permutation $\pi(A)\in S_M$).
  The $\kappa$-direct sum $A[\alpha_1]\hat\oplus A[\alpha_2]$ is then given by
 \[
A[\alpha_1]\hat\oplus A[\alpha_2]=\begin{pmatrix}
1& & & &a & & & &-b\\
 & 1& & -d& & & & & \\
  & &1 &e& & & & & \\
   & & & & &1 & & &c\\
   & & & & & &1 &f &\\
   \end{pmatrix} \in \Gr(5,9).
   \]
Note that this is not totally nonnegative, e.g. $\Delta_{2,3,6,7,9}=-b<0$. This is due to the relocation of the column vectors.
 \end{example}

We now define the notion of \emph{non-crossing} of the permutations.
\begin{definition}\label{def:NC}
Let $\pi^{(p)}= (j^{(p)}_1,\ldots,j^{(p)}_{k_p})$ and $\pi^{(q)}=(j^{(q)}_1,\ldots,j^{(q)}_{k_q})$ be
two permutations in cycle notation. Then we define
\begin{itemize}
\item[(a)] {Two permutations $\pi^{(p)}$ and $\pi^{(q)}$ are \emph{non-crossing}, if the corresponding chord diagrams have no crossing chords between two permutations.}
\item[(b)] Two matrices $A[\alpha_i]$ and $A[\alpha_j]$ for $i\ne j$ are \emph{non-crossing}, if the corresponding permutations $\pi(A[\alpha_i])$ and $\pi(A[\alpha_j])$ are
{non-crossing}.
\end{itemize}
\end{definition}
Note that the matrices $A[\alpha_1]$ and $A[\alpha_2]$ in Example \ref{ex:17} are non-crossing.


An immediate consequence of this definition is the following.
\begin{proposition}\label{lem:crossing}
If $A[\alpha_i]\in\Gr(n_i,m_i)_{\ge0}$ and $A[\alpha_j]\in\Gr(n_j,m_j)_{\ge0}$ for $\alpha_i\ne\alpha_j$ are \emph{not} non-crossing,
then the $\kappa$-direct sum $A[\alpha_i]\hat\oplus A[\alpha_j]$ is not totally nonnegative.
\end{proposition}
\begin{proof}
Recall that the dimension of the totally nonnegative cell can be computed from the formula \eqref{eq:dimA}.
Note that the $\kappa$-direct sum $A[\alpha_i]\hat\oplus A[\alpha_j]$ has an extra crossing.
If $A[\alpha_i]\hat\oplus A[\alpha_j]$ is totally nonnegative, then the number of free parameter (or dimension) in the $\kappa$-direct sum should be
more that the sum of the free parameters in $A[\alpha_i]$ and $A[\alpha_j]$.
\end{proof}

In order to characterize the non-crossing matrices, let us first define the following set $\mathsf{P}(A[\alpha_i])$ for the matrix $A[\alpha_i]\in\text{Gr}(n_i,m_i)$,
\begin{equation*}\label{def:minors}
\mathsf{P}(A[\alpha_i]):=\left\{\Delta_I(A[\alpha_i])\ne0:I\in\mathcal{M}(A[\alpha_i])\right\},
\end{equation*}
where $\mathcal{M}(A[\alpha_i])$ is the matroid of $A[\alpha_i]$, i.e.
\begin{equation*}\label{def:matroid}
\mathcal{M}(A[\alpha_i]):=\left\{I\in\binom{[m_i]}{n_i}:\Delta_I(A[\alpha_i])\ne 0\right\}.
\end{equation*}
Then the following lemma is immediate by the definition of the $\kappa$-direct sum $A[\alpha_i]\hat\oplus A[\alpha_j]$ and using the Laplace expansion for the minors.
\begin{lemma}\label{lem:matP}
\[
\mathsf{P}\left(A[\alpha_i]\hat\oplus A[\alpha_j]\right)=\left\{\Delta_{I}(A[\alpha_i])\Delta_J(A[\alpha_j]):I\in\mathcal{M}(A[\alpha_i]),~J\in\mathcal{M}(A[\alpha_j])\right\}.
\]
\end{lemma}

Let $A[\alpha_i,\alpha_j]$ be the totally non-negative matrix, i.e. $A[\alpha_i,\alpha_j]\in\Gr(n_i+n_j,m_i+m_j)_{\ge 0}$,
whose permutation is given by
\[
\pi(A[\alpha_i,\alpha_j])=\pi(A[\alpha_i])\cdot \pi(A[\alpha_j]).
\]
This implies that the set of the asymptotic solitons in $|y|\gg0$ generated by $A[\alpha_i,\alpha_j]$ is the sum of these solitons generated by
 $A[\alpha_i]$ and $A[\alpha_j]$. Then we can show that Lemma \ref{lem:matP} implies
the following.
\begin{corollary}\label{coro:Aij}
The set of the dominant exponentials in the soliton solution generated by the totally nonnegative matrix
$A[\alpha_i,\alpha_j]$ is the same as that of the solution generated by the matrix $A[\alpha_i]\hat\oplus A[\alpha_j]$.
\end{corollary}
\begin{proof}
 From Lemma \ref{lem:matP}, one can assume that
there exists a minor $\Delta_I(A[\alpha_i])\Delta_J(A[\alpha_j])$ such that the corresponding exponential $E_{\hat I\cup \hat J}$ with $(\hat{I}=\phi(I), \hat{J}=\phi(J))$
by $\phi$ in\eqref{eq:IJ}
is dominant in some asymptotic region of $xy$-plane with large $x^2+y^2$. Then moving $x$-coordinate, the following two cases are possible,
\begin{itemize}
\item[(a)] there exists $I'\in\mathcal{M}(A[\alpha_i])$ with $|I\setminus I'|=1$, so that $E_{\hat I'\cup \hat J}$ dominates over $E_{\hat I\cup \hat J}$,
\item[(b)] there exists $J'\in\mathcal{M}(A[\alpha_j])$ with $|J\setminus J'|=1$, so that $E_{\hat I\cup \hat J'}$ dominates over $E_{\hat I\cup \hat J}$,
\end{itemize}
where we have used the fact that there is only one index change in the minor at the boundary of two dominant regions (more precisely, if $\kappa_k+\kappa_l\ne\kappa_{k'}+\kappa_{l'}$ for $(k,l)\ne(k',l')$, see \cite{CK:08} for the details).
The case (a) implies that there is an asymptotic soliton given in $\pi(A[\alpha_i])$, and the case (b) shows the
existence of the soliton in $\pi(A[\alpha_j])$.
\end{proof}

 Note that these solitons generated by the $\kappa$-direct sum $A[\alpha_i]\hat\oplus A[\alpha_j]$ are singular in general.
Now we can show the following proposition.
\begin{proposition}\label{prop:noncrossing}
Let $A[\alpha_i]\in\Gr(n_i,m_j)_{\ge0}$ and $A[\alpha_j]\in\Gr(n_j,m_j)_{\ge0}$ are non-crossing. Then the $\kappa$-direct sum $A[\alpha_i]\hat\oplus A[\alpha_j]$ becomes totally nonnegative by adjusting the signs in the nonzero entries in $A[\alpha_i]$ and $A[\alpha_j]$.
\end{proposition}
\begin{proof}
Since $A[\alpha_i]$ and $A[\alpha_j]$ are non-crossing, we have
from Lemma \ref{lem:matP} that the number of free parameters of $A[\alpha_i]\hat\oplus A[\alpha_j]$ is just the sum of those in $A[\alpha_i]$ and $A[\alpha_j]$, and the number of free parameters in the totally nonnegative matrix $A[\alpha_i,\alpha_j]$ is the same as that of $A[\alpha_i]\hat\oplus A[\alpha_j]$. These free parameters are given by the nonzero entires of the matrices $A[\alpha_i]$ and $A[\alpha_j]$. Let $a_{k,l}$ be a nonzero entry corresponding to that in $A[\alpha_i]$ or $A[\alpha_j]$, where $\{{k},{l}\}$ is either in $\hat{\mathcal{I}}[\alpha_i]$ or $\hat{\mathcal{I}}[\alpha_j]$.
Then there exists a unique pair $(I,J)\in\mathcal{M}(A[\alpha_i])\times\mathcal{M}(A[\alpha_j])$, so that $\Delta_{\hat I\cup \hat J}(A[\alpha_i]\hat\oplus A[\alpha_j])=\pm\Delta_I(A[\alpha_i])\Delta_J(A[\alpha_j])=\pm a_{k,l}$, where $\hat{I}=\phi(I)$ and $\hat{J}=\phi(J)$. Note also that $\Delta_{\hat{I}\cup\hat{J}}(A[\alpha_i,\alpha_j])=\pm a_{k,l}$.
This determines the signs of all nonzero entries $a_{{k},{l}}$ in $A[\alpha_i]\hat\oplus A[\alpha_j]$. Since $A[\alpha_i,\alpha_j]$ is unique and have the same sets of minors as $A[\alpha_i]\hat\oplus A[\alpha_j]$,
we have that $\Delta_{\hat I\cup \hat{J}}(A[\alpha_i]\hat\oplus A[\alpha_j])>0$ for all $I\in\mathcal{M}(A[\alpha_i])$ and $J\in\mathcal{M}(A[\alpha_j])$.\end{proof}

As the summary of the results in this section, we have the main theorem.
\begin{theorem}\label{thm:main}
Suppose that $A[\alpha_i]$ and $A[\alpha_j]$ are non-crossing for any $1\le i\ne j\le K$. Then we can make the combined matrix $\hat\oplus_{k=1}^KA[\alpha_k]$  totally nonnegative, and we have
\[
\pi(A[\alpha_1,\ldots,\alpha_K])=\prod_{k=1}^K\pi(A[\alpha_k]),
\]
where $A[\alpha_1,\ldots,\alpha_K]$ denotes the totally nonnegative matrix associated to $\hat\oplus_{k=1}^KA[\alpha_k]$, i.e.
$A[\alpha_1,\ldots,\alpha_K]\in\Gr(N,M)_{\ge0}$ with $N=n_1+\cdots+n_K$ and $M=m_1+\cdots+m_K$.
\end{theorem}

This theorem implies that the soliton solution generated by the soliton parameters $\kappa=(\kappa_1,\ldots,\kappa_M)$ with the sorted coordinates $\kappa_m=\kappa_j[\alpha_k]$
($m=\phi(j,k)$) and $A[\alpha_1,\ldots,\alpha_K]\in\Gr(N,M)_{\ge0}$ is regular, that is, it is a KP soliton.

{\begin{remark}
In the combinatorics of totally nonnegative Grassmannian, non-crossing permutations play a fundamental role in totally nonnegativity, as explained in \cite{W:16}. It may be interesting to establish a precise connection between our results and their results.
\end{remark}}


\section{Regular solitons for the Boussinesq equation}\label{sec:Boussinesq}
Based on the previous section, we give the detailed study of the real regular soliton solutions of the Boussinesq equation.
\subsection{The Boussinesq equation from the KP theory}
It is well known that the Boussinesq equation is given by the 3-reduction of the KP hierarchy, i.e. $L^3=(L^3)_{\ge0}=B_3$. We write
\begin{equation*}\label{B3}
L^3= B_{3}:=\partial^{3}+\frac{3}{2}u\partial+\frac{3}{4}u_{x}+w,
\end{equation*}
where $u=2u_2$ and $w=3u_3+\frac{3}{2}u_{2,x}$ with the Lax operator \eqref{eq:LaxOP}.
Then the Boussinesq equation is obtained by the Lax equation,
\begin{equation*}\label{BH}
\frac{\partial L^{3}}{\partial t_{2}}=\frac{\partial B_{3}}{\partial t_{2}}=[B_{2}, B_{3}],
\end{equation*}
where $B_2=(L^2)_{\ge0}=\partial^2+u$. This gives
\begin{eqnarray}\label{eq:BOUT}
\left\{ \begin{array}{ll}\vspace{1ex}\begin{aligned}
 & \frac{\partial u}{\partial t_2}=\frac{4}{3} w_x,\\
 & \frac{\partial w}{\partial t_2}=-\frac{1}{4} u_{xxx}-\frac{3}{2}u u_{x}.
\end{aligned}\end{array} \right.
\end{eqnarray}
The functions $u$ and $w$ can be represented by the $\tau$-function,
\begin{equation}\label{eq:Btau}
u=2\frac{\partial^2}{\partial x^2}\ln \tau,\qquad \text{and}\qquad w=\frac{3}{2}\frac{\partial^2}{\partial x\partial  t_2}\ln \tau.
\end{equation}
Eliminating $w$ in \eqref{eq:BOUT}, we have
\begin{equation}\label{eq:BOU}
3u_{t_2t_2}+(3u^2+u_{xx})_{xx}=0.
\end{equation}
Notice that this is not a standard form of the Boussinesq equation. Also note that this equation does not have a soliton solution with the vanishing boundary condition, i.e. $u\to 0$ as $|x|\to\infty$. In order to obtain a regular soliton (exponential) solution,
we need to impose a non-vanishing boundary condition. Assuming $u\to u+u_0$, \eqref{eq:BOU} becomes
\begin{align}\label{eq:SB}
3{u}_{t_2t_2}+6u_0{u}_{xx}+(3{u}^{2}+{u}_{xx})_{xx}=0.
\end{align}
In terms of the $\tau$-function in \eqref{eq:Btau}, this shift of $u$ implies
\begin{equation*}\label{eq:u-shift}
\tau\quad\longrightarrow\quad e^{\frac{1}{4}u_0x^2}\tau.
\end{equation*}
One can easily check that \eqref{eq:SB} admits a soliton solution when $u_0<0$. This shifted form of the Boussinesq equation can be also obtained by the coordinate change in the KP equation,
\begin{eqnarray}\label{eq:t-shift}
\left\{ \begin{array}{ll}\vspace{1ex}\begin{aligned}
 & x~\to~x+c_{1}t_3,\\
 & t_2~\to~t_2+c_2t_3,\\
 & t_3 ~\to~ {t_3},
\end{aligned}\end{array} \right. \qquad\text{which gives}\qquad
\left\{ \begin{array}{ll}\vspace{1ex}\begin{aligned}
& \frac{\partial}{\partial x}~\to~\frac{\partial}{\partial {x}},\quad\quad\frac{\partial}{\partial t_2}~\to~\frac{\partial}{\partial {t_2}},\\
 & \frac{\partial}{\partial t_3}~\to~\frac{\partial}{\partial t_3}+c_{1}\frac{\partial}{\partial {x}}+c_{2}\frac{\partial}{\partial {t_2}},
\end{aligned}\end{array} \right.
\end{eqnarray}
where $c_{1}, c_{2}$ are arbitrary constants. After the change of coordinates (\ref{eq:t-shift}), the KP equation (\ref{eq:KP}) can be expressed by
\begin{equation}\label{eq:sKP}
-4u_{xt}-4c_1u_{xx}-4c_2u_{xt_2}+(3u^2+u_{xx})_{xx}+3u_{t_2t_2}=0.
\end{equation}
Then we consider a stationary solution in $t_3$ (the 3-reduction), that is, $u_{t_3}=0$. By choosing $c_{2}=0$ and $c_{1}=-\frac{3}{2}u_0$, the equation (\ref{eq:sKP}) becomes (\ref{eq:SB}). As was shown in the previous section, the change of coordinates now leads to the versal deformation of the spectral curve, i.e.
\begin{equation}\label{eq:Bcurve}
\Phi_3(\kappa,\alpha)=\kappa^3-c_1\kappa-\alpha=0.
\end{equation}
Note here if $c_1>0$, we have three real roots for some $\alpha\in\R$, which gives a real exponential basis for the soliton solution. Then the soliton solutions of the Boussinesq equation,
\begin{equation}\label{eq:SBoussinesq}
3u_{t_2t_2}-4c_1u_{xx}+(3u^2+u_{xx})_{xx}=0,
\end{equation}
are generated by the $\tau$-function in the following form,
\begin{equation}\label{eq:B-tau}
\tau(\t)=\text{Wr}(f_1,\ldots,f_N)(\t),
\end{equation}
where $f_j(\t)$ is given by a linear combination of the exponential functions \eqref{eq:E}. In this section,
we consider the Boussinesq equation in the following standard form,
\begin{equation}\label{eq:B-standard}
u_{t_2t_2}-c_0^2u_{xx}+\frac{1}{3}\left(3u^2+u_{xx}\right)_{xx}=0,
\end{equation}
where we have taken $c_1=3c_0^2/4$ $(\text{we assume}~c_0>0)$.

\subsection{One soliton solution of the Boussinesq equation}
A soliton solution of the Boussinesq equation \eqref{eq:B-standard} as one soliton solution of the KP equation with $\ell=3$ is given by
\begin{equation}\label{eq:1-solB}
u(x,t_2)=\frac{(\kappa_i-\kappa_j)^2}{2}\sech^2\frac{\kappa_i-\kappa_j}{2}\left(x+(\kappa_i+\kappa_j)t_2+x_0\right),
\end{equation}
where $(\kappa_i,\kappa_j)$ is a pair of roots of the curve \eqref{eq:Bcurve}. The amplitude and the velocity of the $[i,j]$-soliton is then given by
\begin{equation*}\label{eq:av}
a_{[i,j]}=\frac{1}{2}(\kappa_i-\kappa_j)^2,\qquad v_{[i,j]}=-(\kappa_i+\kappa_j).
\end{equation*}
Since the parameters $(\kappa_i,\kappa_j)$ satisfy the curve \eqref{eq:Bcurve}, the amplitude and the velocity have the relations.
Note first that the roots $\{\kappa_1,\kappa_2,\kappa_3\}$ of the spectral curve \eqref{eq:Bcurve} satisfy the symmetric polynomials,
\[
\sum_{j=1}^3\kappa_j=0,\qquad\sum_{j<k}\kappa_j\kappa_k=-c_1,\qquad\prod_{j=1}^3\kappa_j=\alpha.
\]
Then from the first two equations, we see that any pair of two roots $(\kappa_i,\kappa_j)$ of \eqref{eq:Bcurve} satisfies the elliptic curve,
\begin{equation}\label{eq:Belliptic}
\left(\kappa_i+\kappa_j\right)^2+\frac{1}{3}\left(\kappa_i-\kappa_j\right)^2=\frac{4}{3}c_1=c_0^2,\quad\text{that is,}\quad v_{[i,j]}^2+\frac{2}{3}a_{[i,j]}=c_0^2.
\end{equation}
We note here that there are two groups of solitons having opposite propagating directions. The details can be computed as follows.
The elliptic curve \eqref{eq:Belliptic} and the cubic curve \eqref{eq:Bcurve} are illustrated in Figure \ref{fig:elliptic}.
\begin{figure}[H]
  \begin{minipage}[t]{0.45\linewidth}
  \centering
  \includegraphics[height=6cm,width=7cm]{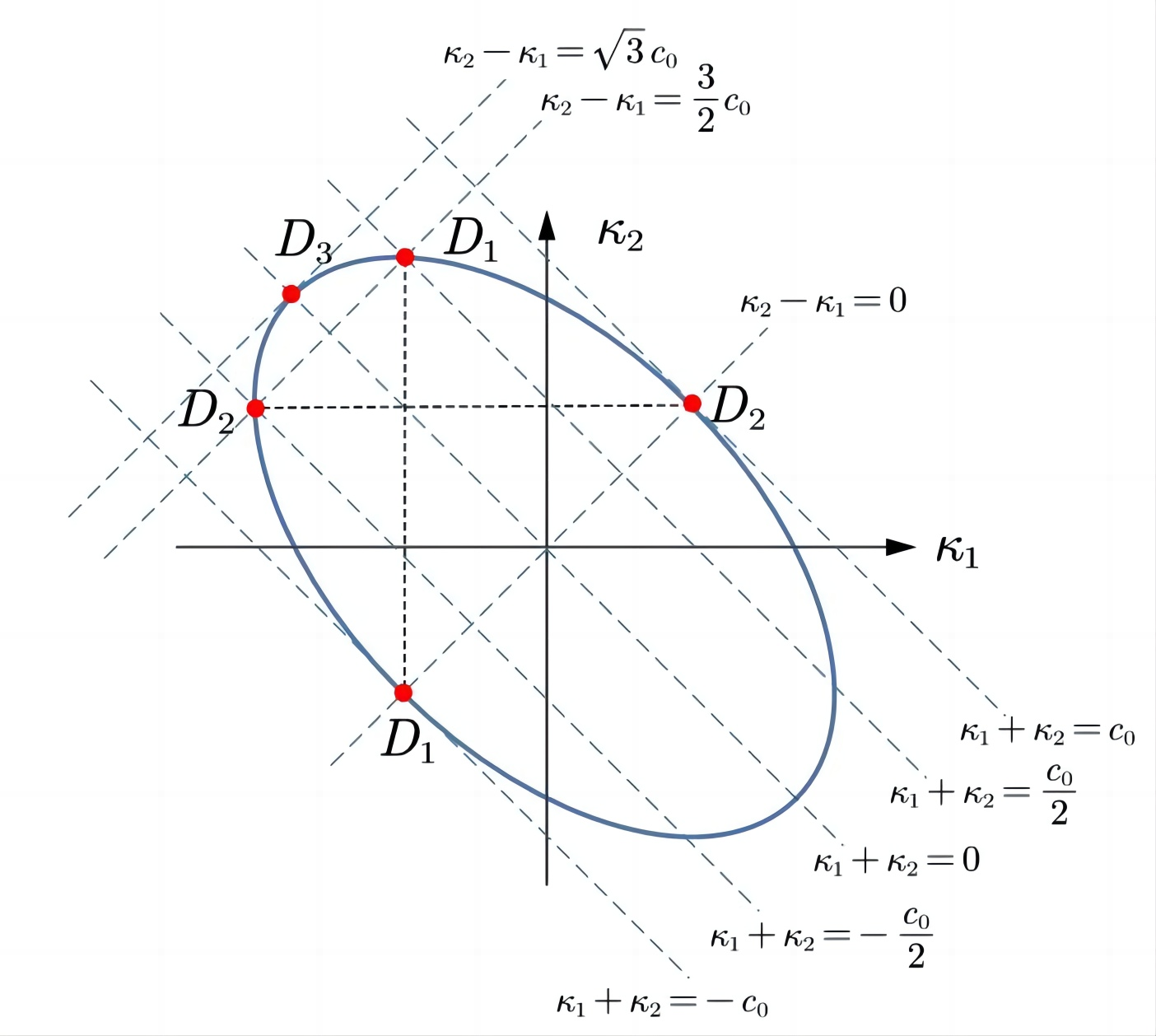}
  \end{minipage}
  \hskip0.5cm
  \begin{minipage}[t]{0.45\linewidth}
  \centering
  \raisebox{0.5cm}{\includegraphics[height=4.5cm,width=7cm]{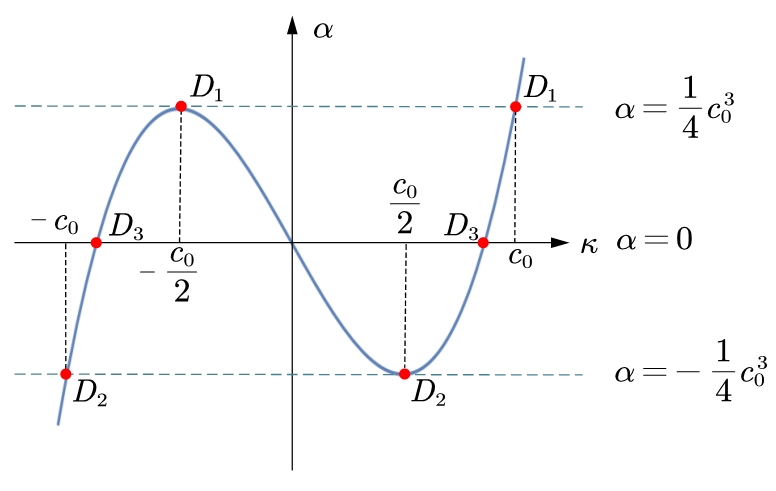}}
  \end{minipage}
  \caption{The left figure is the elliptic curve given by \eqref{eq:Belliptic}, and the right one is the spectral curve $\kappa^3-c_1\kappa=\alpha$.  \label{fig:elliptic}}
\end{figure}
\noindent
The three roots $(\kappa_1<\kappa_2<\kappa_3)$ at the points $D_1, D_2$ and $D_3$ in Figure \ref{fig:elliptic} are
\[
D_1:\quad \left(-\frac{1}{2} c_0,-\frac{1}{2} c_0, c_0\right),\qquad D_2:\quad \left(-c_0,\frac{1}{2} c_0,\frac{1}{2} c_0\right),\qquad D_3:\quad \left(-\frac{\sqrt 3}{2}c_0,0, \frac{\sqrt{3}}{2}c_0\right).
\]
Then we note that there are three types of solitons:
\begin{itemize}
\item[(a)] the right propagating solitons with
\[
\frac{1}{2}c_0<v_{[1,2]}<c_0,\qquad 0<a_{[1,2]}<\frac{9}{8}c_0^2,
\]
\item[(b)] the slow propagating solitons with
\[
-\frac{1}{2}c_0<v_{[1,3]}<\frac{1}{2}c_0,\qquad \frac{9}{8}c_0^2<a_{[1,3]}\le \frac{3}{2}c_0^2,
\]
\item[(c)]
the left propagating solitons with
\[
-c_0<v_{[2,3]}<-\frac{1}{2}c_0,\qquad 0<a_{[2,3]}<\frac{9}{8}c_0^2.
\]
\end{itemize}
Note that the amplitude $a_{[i,j]}$ of the soliton is limited by $3c_0^2/2$.
\begin{remark}
We remark that the linear wave of the Boussinesq equation satisfies the dispersion relation,
\[
\omega^2=c_0^2k^2 +\frac{1}{3}k^4,\qquad\text{which gives}\quad v=\frac{\omega}{k}=\pm c_0\sqrt{1+\frac{1}{3c_0^2}k^2}.
\]
This is derived from the linear part of \eqref{eq:B-standard} with $u=e^{ikx-i\omega t_2}$.
Note here that $|v|>c_0$, that is, the linear waves propagate faster than the solitons (soliton resolution).
\end{remark}

\subsection{Multi-soliton solutions}
We now construct a general soliton solution of the Boussinesq equation \eqref{eq:B-standard} by taking several different values of $\alpha$ in the spectral curve \eqref{eq:Bcurve}.
For a proper choice of $\alpha$, we have three real distinct roots as shown below.
\vskip-0.5cm
\begin{figure}[H]
  \begin{minipage}[t]{1\linewidth}
  \centering
  \includegraphics[height=3.5cm,width=8cm]{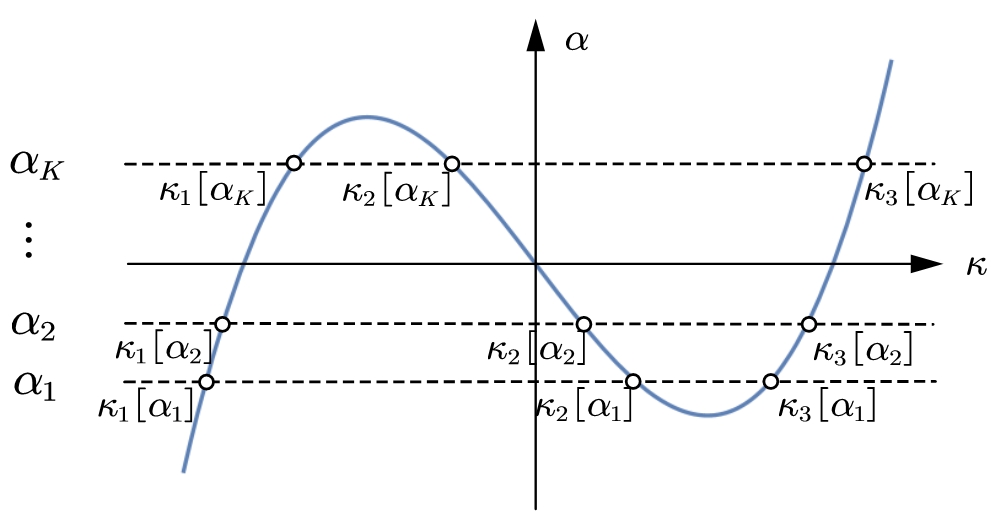}
  \caption{For each $\alpha=\alpha_j$, there are three roots labeled as $(\kappa_1[\alpha_j]<\kappa_2[\alpha_j]<\kappa_3[\alpha_j])$.}
  \end{minipage}
\end{figure}
\vskip-0.5cm
We have the following two cases for the index set $\mathcal{I}[\alpha]\subset [3]$,
$\mathcal{I}[\alpha]\subset [3]=\{1,2,3\}$ can be chosen in the following two cases:
\begin{itemize}
\item[(a)] $\mathcal{I}[\alpha]=2$, i.e. we take two roots $(\kappa_i[\alpha]<\kappa_j[\alpha])$. There are three different choices, and the soliton solution
is given by \eqref{eq:1-solB} with $A=(1,a)\in \Gr(1,2)_{\ge 0}$ (these are the solitons discuss above).
\item[(b)] $\mathcal{I}[\alpha]=3$, i.e. we take all three roots.
The soliton solution in this case has a \emph{resonant} interaction with three solitons $[1,2]$-, $[1,3]$- and $[2,3]$-types.
There are two types of resonant solution with $A[\alpha]=(1,a_1,a_2)\in\Gr(1,3)_{\ge0}$ or $A[\alpha]=\begin{pmatrix}1&0&-a_1\\0&1&a_2\end{pmatrix}\in\Gr(2,3)_{\ge0}$. These soliton solutions form a $Y$-shape resonant structure as shown in the figure below.
\end{itemize}
The permutation diagrams corresponding to these solutions are given by
\begin{figure}[H]
  \begin{minipage}[t]{0.33\linewidth}
  \centering
  \includegraphics[height=1.5cm,width=2.5cm]{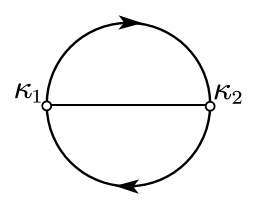}
  \end{minipage}%
 \hskip-1cm
  \begin{minipage}[t]{0.33\linewidth}
  \centering
  \includegraphics[height=1.5cm,width=3cm]{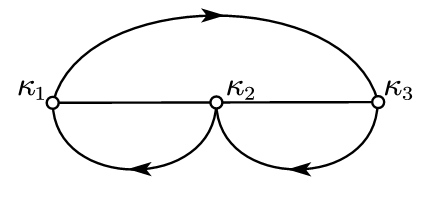}
  \end{minipage}
  \hskip-0.5cm
  \begin{minipage}[t]{0.33\linewidth}
  \centering
  \includegraphics[height=1.5cm,width=3cm]{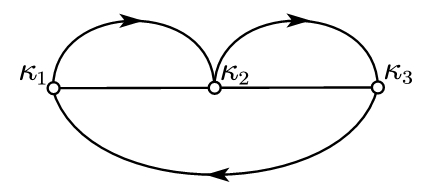}
  \end{minipage}
  \caption{The chord diagrams $\pi(A[\alpha])$ for 1-soliton and Y-solitons.}
\end{figure}

\vskip-0.2cm
A general soliton solutions of the Boussinesq equation can be constructed by a $\kappa$-direct sum of non-crossing matrices $A[\alpha_k]$ for $k=1,\ldots,K$ with some $K$.
Then from Theorem \ref{thm:main}, we have the following theorem.
\begin{theorem}\label{thm:B}
The $\tau$-function of any real regular soliton solution of the Boussinesq equation \eqref{eq:SB} can be generated by
one of the following three cases with the soliton parameters $\kappa\in\R^M$ and $A\in\text{Gr}(N,M)_{\ge 0}$.
\begin{itemize}
\item[(I)] We take the following sets of the roots for $\alpha_k$ $(1\le k\le K=1+\gamma_1+\gamma_2$) with $\{\kappa_i[\alpha_1]:i=1,2,3\}$ and
\[
\{\kappa_1[\alpha_{1+j}],\kappa_2[\alpha_{1+j}]:j=1,\ldots,\gamma_1\}\cup\{\kappa_2[\alpha_{\gamma_1+1+l}],\kappa_3[\alpha_{\gamma_1+1+l}]:l=1,\ldots,\gamma_2\}.
\]
Then, we have the sorted soliton parameter $\kappa=(\kappa_1,\ldots,\kappa_M)$ with $M=2K+1=3+2(\gamma_1+\gamma_2)$,
\begin{align*}
\kappa_1&=\kappa_1[\alpha_1]<\kappa_1[\alpha_2]<\cdots<\kappa_1[\alpha_{\gamma_1+1}]<\kappa_2[\alpha_{\gamma_1+1}]<\cdots<\kappa_2[\alpha_2]<\\
&<\kappa_2[\alpha_1]<\kappa_2[\alpha_{\gamma_1+2}]<\cdots<\kappa_2[\alpha_K]<\kappa_3[\alpha_K]<\cdots<\kappa_3[\alpha_{\gamma_1+2}]<\kappa_3[\alpha_1]=\kappa_M.
\end{align*}
\begin{figure}[H]
  \begin{minipage}[t]{1\linewidth}
  \centering
  \includegraphics[height=3.5cm,width=11cm]{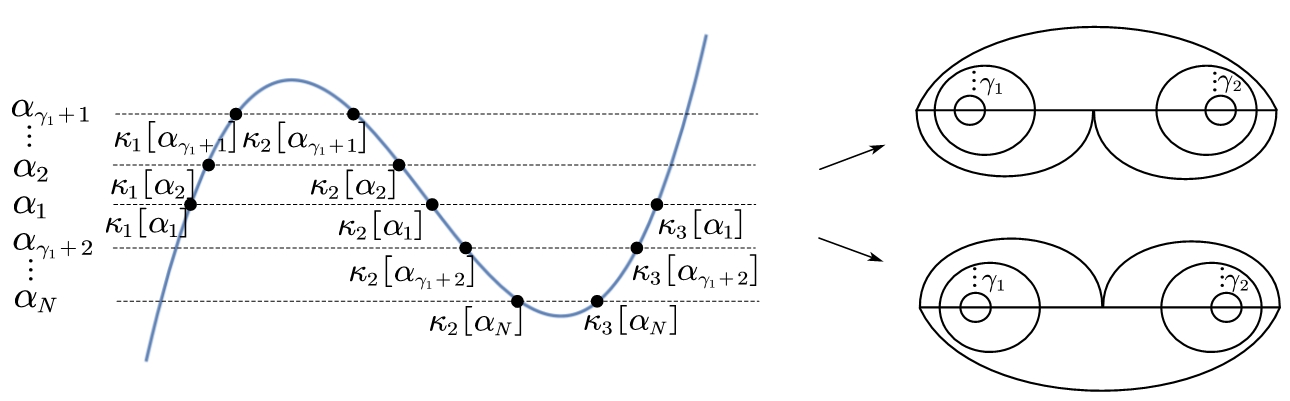}
  \caption{{The choices of roots in the case (I), and the permutations of the soliton solutions having one $Y$-soliton.}}
  \end{minipage}
\end{figure}
Then, there are two cases in the choice of $\kappa$-direct sum $A:=\hat\oplus_{k=1}^KA[\alpha_k]$, where $A[\alpha_1]\in\Gr(1,3)$ or $A[\alpha_1]\in \Gr(2,3)$ while $A[\alpha_j]\in\Gr(1,2)$ for $j=2,\ldots,K$.
\begin{itemize}
\item[(a)]
The chord diagram shown in the top right shows $A[\alpha_1]\in\text{Gr}(1,3)$, and the corresponding derangement is
\[
\pi(A)=(1,M,2\gamma_1+2)\prod_{k=1}^{\gamma_1}(k+1,2\gamma_1+2-k)\prod_{l=1}^{\gamma_2}(2\gamma_1+2+l,M-l).
\]
In this case, we have $A\in\Gr(N,M)_{\ge0}$ with $N=K$ and $M=2K+1$.
\item[(b)]
The $\pi(A)$ for the bottom right is the same except the identification of $\pi(A[\alpha_1]))$,
\[
\pi(A[\alpha_1])=(1,2\gamma_1+2,M).
\]
In this case, we have $A\in\text{Gr}(N,M)$ with $N=K+1$ and $M=2K+1$.
\end{itemize}
\item[(II)] We take just two roots of $\Phi_3(\kappa,\alpha_i)=0$ for each $i$, and take $i=1,\ldots,N$, $N=1+\gamma_1+\gamma
_2$, i.e. the soliton parameter $\kappa=(\kappa_1,\ldots,\kappa_M)$ with $M=2N=2(1+\gamma_1+\gamma_2)$.
There are two cases as shown in the figures below.
\vskip-0.5cm
\begin{figure}[H]
  \begin{minipage}[t]{1\linewidth}
  \centering
  \includegraphics[height=3cm,width=11cm]{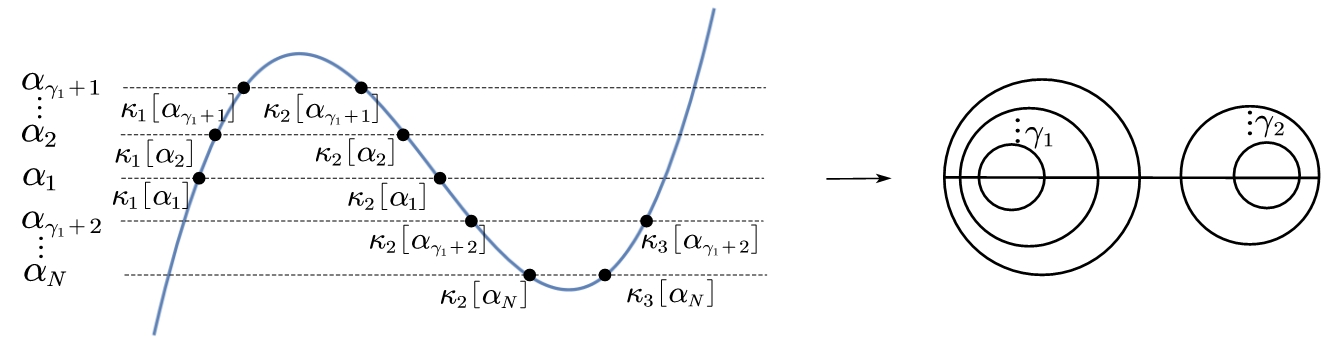}
  \caption{{The choice of the roots for the case (II) and the permutation of soliton solution consisting of two groups propagating opposite directions.}\label{fig:2G}}
  \end{minipage}
    \begin{minipage}[t]{1\linewidth}
  \centering
  \includegraphics[height=3cm,width=11cm]{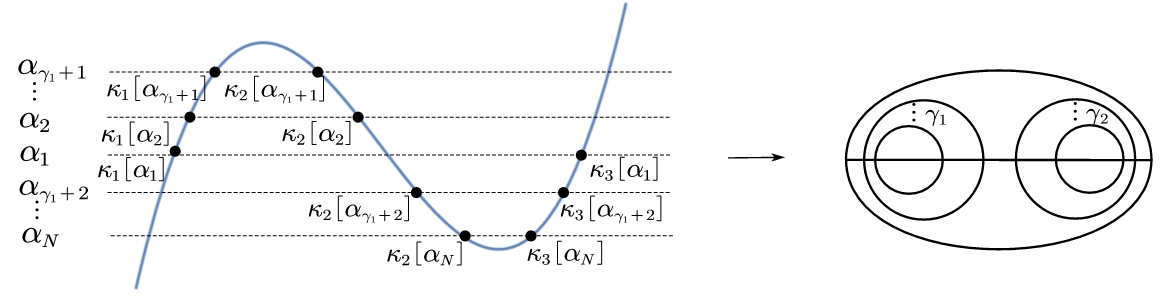}
  \caption{{The soliton solution of the type in Figure \ref{fig:2G} with a slow moving large amplitude soliton.}}
  \end{minipage}
\end{figure}
\end{itemize}
\end{theorem}

\subsection{Discussions}
In this section, we studied the (good) Boussinesq equation and classified the general solutions of the equation.
We hope that our results provide a model of bidirectional soliton gas as discussed recently in \cite{BD:24}.
We also found that there exists at most one resonant solution (Y-shape soliton) or one slow propagating soliton with large amplitude for the regular solitons
(see \cite{RS:17}).
It might be interesting to discuss the effect of those special solitons among two groups of counter propagating solitons.
However, it may not be so physical if the (good) Boussinesq equation is used for a shallow water wave model.
This is because that a \emph{larger} soliton has a \emph{slower} velocity, and even that the largest soliton has the zero velocity. {It was shown in \cite{Z:02} that these slow solitons are unstable under a small perturbation (also see \cite{BS:88})}.

Although we did not discuss the details of the higher reductions in this paper, one can have $\ell-1$ different groups having the distinct velocities for the $\ell$-reduction.
Their interaction properties such as the phase shifts are different when they interact with other solitons from different groups.
It is also interesting to note that one can give additional characters, like amplitudes and velocities, by a spectral curve with different deformation parameters.
We will report the details on physical applications of the $\ell$-reductions of the KP equation in a future communication.

\section{Vertex operator construction of the KP solitons under the $\ell$-reductions}\label{sec:VO}
It is well known that applying the vertex operator \cite{DKJM:83, Na:23}, one can construct the soliton solutions of the KP hierarchy.
However, as far as we know, the regularity of those solutions have not been discussed.
In this section, we determine the conditions to the vertex operators, so that these operators generate the regular soliton solutions under the $\ell$-reductions.

The vertex operator is defined by the following form with arbitrary parameters $\{p_i,q_j: 1\le i\le P,~1\le j\le Q\}$ for some positive integers $P$ and $Q$,
\begin{equation*}\label{eq:Xvertex}
X(p_i, q_j)=e^{\varphi(\t,p_i,q_j)}e^{-\varphi(\tilde\partial,p_i^{-1},q_j^{-1})},\quad\text{where}\quad \varphi(\t,p_i,q_j):=\sum_{n=1}^\infty(q_j^n-p_i^n)t_n,
\end{equation*}
{where $\tilde{\partial}=(\partial_{x},\frac{1}{2}\partial_{t_{2}},\cdots)$}. The following lemma is well-known and easy to show.
\begin{lemma}\label{lem:X}
The vertex operator $X(p_{i},q_j)$ satisfies
\begin{equation*}
X(p_{1},q_{1})X(p_{2},q_{2})=\frac{(p_{1}-p_{2})(q_{1}-q_{2})}{(p_{1}-q_{2})(q_{1}-p_{2})}: X(p_{1},q_{1})X(p_{2},q_{2}):,
\end{equation*}
where the normal ordering symbol $:\cdot:$ implies to move differential operators $\partial_{n}$ to the right. It then follows the following properties
\begin{align*}
&X(p_{1}, q_{1})X(p_{2}, q_{2})=X(p_{2}, q_{2})X(p_{1}, q_{1}), \qquad\text{if}\quad (p_1-q_2)(q_1-p_2)\ne 0,\\
&X(p_{1}, q_{1})X(p_{2}, q_{2})=0,\quad \quad\quad\text{if $p_{1}=p_{2}$ or $q_{1}=q_{2}$}.
\end{align*}
\end{lemma}

\subsection{The vertex operators on $\Gr(n,m)_{\ge0}$}
In the present paper, the parameters $\{p_i,q_j\}$ are determined by the roots of the spectral curve $\Phi_\ell(\kappa,\alpha)=0$ for each $\alpha$, i.e. $p_i=\kappa_{i'}[\alpha]$ and $q_j=\kappa_{j'}[\alpha]$ for some $i', j'\in \mathcal{I}[\alpha]$. As in the previous section, we take $2\le |\mathcal{I}[\alpha]|=m\le \ell$, and consider $1\le n\le m-1$.
We take an (irreducible) element $A[\alpha]\in\Gr(n,m)_{\ge 0}$. Let $\mathcal{I}_P[\alpha]$ and $\mathcal{I}_Q[\alpha]$ be
the sets of pivots and non-pivots of the matrix $A[\alpha]$, i.e. $\mathcal{I}[\alpha]=\mathcal{I}_P[\alpha]\sqcup\mathcal{I}_Q[\alpha]$.
We use the following labels for the nonzero elements in $A[\alpha]$,
\begin{align*}
\mathcal{I}_P[\alpha]&=\{i_1<i_2<\cdots< i_n\},\\
\mathcal{I}_Q[\alpha]&=\bigcup_{k=1}^n\{j_1^{(k)}<j_2^{(k)}<\cdots<j_{q_k}^{(k)}\},
\end{align*}
where $j_l^{(k)}$ is the column index of the nonzero element in the $k$-th row and {$j_{l}^{(k)}$-th column}. That is, the nonzero elements except the pivots in the matrix $A[\alpha]$ are
\begin{equation*}\label{eq:Aij}
\left\{\,a_{k,j^{(k)}_l}: \,1\le l\le q_k,~1\le k\le n\,\right\}.
\end{equation*}
Following \cite{K:24}, we give the following order for the indices $j_l^{(k)}$,
\begin{equation}\label{eq:order}
j^{(k)}_l\quad\longleftrightarrow\quad \tilde g_{k-1}+l,\qquad (1\le l\le q_k,~\, 1\le k\le n),
\end{equation}
where $\tilde g_k:=q_1+\cdots+q_k$ with $\tilde g_0=0$ and $\tilde g_n:=\tilde g$. That is, the total number of the nonzero elements in $A[\alpha]$ is $\tilde g$.

Let us define a \emph{positive} $n\times (m-n)$ matrix  $B[\alpha]=(b_{k,j^{(k)}_l})$ as
\begin{equation*}\label{eq:Bij}
b_{k,j^{(k)}_l}=a_{k,j^{(k)}_l}\prod_{p\ne k}\frac{\kappa_{i_p}-\kappa_{j^{(k)}_l}}{\kappa_{i_p}-\kappa_{i_k}}>0.
\end{equation*}

Then we have the following lemma.
\begin{lemma}\label{eq:sign}
The signs of nonzero entries $a_{k,j^{(k)}_l}$ in $A[\alpha]\in \Gr(n,m)_{\ge 0}$ is determined by the positivity conditions $b_{k,j^{(k)}_l}>0$, i.e.
\[
\text{sgn}(a_{k,j^{(k)}_l})=\text{sgn}\left(\prod_{p\ne k}\frac{\kappa_{i_p}-\kappa_{j^{(k)}_l}}{\kappa_{i_p}-\kappa_{i_k}}\right).
\]
\end{lemma}
\begin{proof}
For the pair $(i_k,j^{(k)}_l)$ of the pivot and the non-pivot indices,  let $(\mu, \nu)$ be the numbers given by
\[
\mu:=\left|\{p: i_p>i_k\}\right|,\qquad \nu:=\left|\{p:i_p>j^{(k)}_l\}\right|.
\]
Then we have
\[
\text{sgn}\left(\prod_{p\ne k}\frac{\kappa_{i_p}-\kappa_{j^{(k)}_l}}{\kappa_{i_p}-\kappa_{i_k}}\right)=(-1)^{\mu+\nu}.
\]
That is, we claim that $b_{k,j^{(k)}_l}>0$ implies $\text{sgn}(a_{k,j^{(k)}_l})=(-1)^{\mu+\nu}$.
This is shown by noting that there is a unique index set $J\in\mathcal{M}(A[\alpha])$ such that
\begin{equation*}
\Delta_{J}(A[{\alpha}])=\left|
\begin{array}{cccc}
I_{n-1-\mu} &  &  & \\
 &   & a_{k,j^{(k)}_l}&\\
 & I_{\mu-\nu} &  & \\
 &  &  & I_{\nu}\\
\end{array}
\right|=(-1)^{\mu-\nu}a_{k,j^{(k)}_l},
\end{equation*}
where $I_{r}$ denotes the identity matrix of $r\times r$. Since $A[\alpha]\in \Gr(n,m)_{\ge 0}$, $\Delta_J(A[\alpha])>0$. This proves the lemma.
\end{proof}

With the positive matrix $B[\alpha]$,
we define our vertex operator in the following form similar to that given in \cite{Na:23},
\begin{equation}\label{eq:VO}
V[\alpha]:=\exp\left(\sum_{k=1}^n\sum_{l=1}^{q_k}b_{k,j^{(k)}_l}X(\kappa_{i_k},\kappa_{j^{(k)}_l})\right)=\prod_{k=1}^{n}\prod_{l=1}^{q_{k}}\left(1+b_{k,j^{(k)}_l}X(\kappa_{i_k},\kappa_{j^{(k)}_l})\right).
\end{equation}
{Note here that the parameters $(\kappa_{i},\kappa_{j})$ in $X(\kappa_{i},\kappa_{j})$ are from $(i\in\mathcal{I}_P[\alpha], j\in\mathcal{I}_Q[\alpha])$, so that these $X(\kappa_{i},\kappa_{j})$ satisfy the condition in Lemma \ref{lem:X}.
Then we have the following proposition.}
\begin{proposition}\label{prop:Vsoliton}
The $\tau$-function in \eqref{eq:tau-alpha} with the soliton parameters
$(\kappa[\alpha]\in\R^m, A[\alpha]\in\Gr(n,m)_{\ge0})$ can be expressed by
\begin{align*}\label{eq:tauV}
\frac{\tau(\t,\alpha)}{E_{I_0}(\t,\alpha)}=V[\alpha]\cdot 1,
\end{align*}
where $I_0$ is the pivot index set of $A[\alpha]$.
\end{proposition}
\begin{proof}
We first note (Theorem 3.10 in \cite{K:24}) that we have $\tau(\t)/E_{I_0}(\t)=\tilde\vartheta_{\tilde g}(z;\tilde\Omega)$, called the $M$-theta function defined by
\begin{align*}
\tilde\vartheta_{\tilde g}(z;\tilde\Omega)&=\sum_{{\bf m}\in\{0,1\}^{\tilde g}}\exp 2\pi i\left(\sum_{j<k}^{\tilde g}m_jm_k\tilde\Omega_{j,k}+\sum_{j=1}^{\tilde g}m_jz_j\right)\\
&=1+\sum_{j=1}^{\tilde g}e^{2\pi iz_j}+\sum_{k<l}^{\tilde g}e^{2\pi i\tilde\Omega_{k,l}}e^{2\pi i(z_k+z_l)}+\cdots+e^{2\pi i\sum_{k<l}\tilde\Omega_{k,l}}e^{2\pi i\sum_{j=1}^{\tilde g}z_j},
\end{align*}
where $2\pi i z_j=\phi_j(\t)+\phi^0_j,~(1\le j\le \tilde g)$ with
\[
\phi_{\tilde g_{k-1}+l}(\t)=\varphi(\t,\kappa_{i_k},\kappa_{j_l^{(k)}}),\quad \phi^0_{\tilde g_{k-1}+l}=\ln \left(b_{k,j_l^{(k)}}\right),\qquad (1\le l\le q_k,~\, 1\le k\le n),
\]
 and
\[
\exp(2\pi i\tilde{\Omega}_{\tilde{g}_{k-1}+l, \tilde{g}_{k'-1}+r})=\frac{(\kappa_{i_{k}}-\kappa_{i_{k'}})(\kappa_{j_{l}^{(k)}}-\kappa_{j_{r}^{(k')}})}{(\kappa_{i_{k}}-\kappa_{j_{r}^{(k')}})(\kappa_{j_{l}^{(k)}}-\kappa_{i_{k'}})},
\]
where we have used the order \eqref{eq:order}.
Then for $z_j$ with $j=\tilde g_{k-1}+l$, the $e^{2\pi iz_j}$ in $\tilde\vartheta_{\tilde g}(z,\tilde\Omega)$ becomes
\[
e^{2\pi i z_j}=b_{k,j_l^{(k)}}e^{\varphi(\t,\kappa_{i_k},\kappa_{j^{(k)}_l})}=b_{k,j_l^{(k)}}X(\kappa_{i_k},\kappa_{j_l^{(k)}})\cdot 1.
\]
The product term $e^{2\pi i\tilde\Omega_{j,l}}e^{2\pi i(z_j+z_l)}$ with $j=\tilde g_{k-1}+r$ and $l=\tilde g_{k'-1}+r'$ gives
\begin{align*}
e^{2\pi i\tilde\Omega_{j,l}}e^{2\pi i(z_j+z_l)}&=\frac{(\kappa_{i_{k}}-\kappa_{i_{k'}})(\kappa_{j_{r}^{(k)}}-\kappa_{j_{r'}^{(k')}})}{(\kappa_{i_{k}}-\kappa_{j_{r'}^{(k')}})(\kappa_{j_{r}^{(k)}}-\kappa_{i_{k'}})}
b_{k,j^{(k)}_r}b_{k',j_{r'}^{(k')}}e^{\varphi(\t,\kappa_{i_k},\kappa_{j_r^{(k)}})+\varphi(\t,\kappa_{i_{k'}},\kappa_{j_{r'}^{(k')}})}\\
&=b_{k,j^{(k)}_r}b_{k',j_{r'}^{(k')}}X(\kappa_{i_k},\kappa_{j_r^{(k)}})X(\kappa_{i_{k'}},\kappa_{j_{r'}^{(k')}})\cdot 1,
\end{align*}
where we have used Lemma \ref{lem:X}.

The higher products can be calculated in the similar way. One should note that some of the terms $e^{2\pi i\tilde\Omega_{j,l}}$ vanish, when $i_k=i_{k'}$ or $j_r^{(k)}=j_{r'}^{(k')}$
(see Lemma \ref{lem:X}).
\end{proof}
Since the solution of the KP equation is given by \eqref{eq:u-tau}, i.e. $u(\t)=2(\ln\tau(\t))_{xx}$, we consider the $\tau$-function as in the form $\tau(\t,\alpha)/E_{I_0}(\t,\alpha)$,
and we write
\begin{equation}\label{eq:tau-V}
\tau(\t,\alpha)=V[\alpha]\cdot 1.
\end{equation}

\subsection{The vertex operator construction of the KP solitons under the $\ell$-reduction}
Here we consider a $\tau$-function generated by several different $\alpha$'s, that is, for some $K>1$,
\begin{equation}\label{eq:tau-VO}
\tau(\t;\alpha_1,\ldots,\alpha_K):=V[\alpha_1]\cdot V[\alpha_2]\cdots V[\alpha_K]\cdot 1,
\end{equation}
where $V[\alpha_j]$ are given by \eqref{eq:VO} with the positive matrices $B[\alpha_j]$ determined by non-crossing matrices $A[\alpha_j]\in\Gr(n_j,m_j)_{\ge 0}$, respectively.

Let $(k[\alpha_r],j^{(k)}_l[\alpha_r])$ be the indices of the pivot and non-pivot columns in the $k$-th row of the nonzero entries in $A[\alpha_r]$. Then we have
\begin{align*}
&X(\kappa_{k[\alpha_r]},\kappa_{j^{(k)}_l[\alpha_r]})\cdot X(\kappa_{k'[\alpha_s]},\kappa_{j^{(k')}_{l'}[\alpha_s]})\\
=&
\frac{(\kappa_{k[\alpha_r]}-\kappa_{k'[\alpha_s]})(\kappa_{j^{(k)}_l[\alpha_r]}-\kappa_{j^{(k')}_{l'}[\alpha_s]})}{(\kappa_{k[\alpha_r]}-\kappa_{j^{(k')}_{l'}[\alpha_s]})
(\kappa_{j^{(k)}_l[\alpha_r]}-\kappa_{k'[\alpha_s]})}\,:X(\kappa_{k[\alpha_r]},\kappa_{j^{(k)}_l[\alpha_r]})\cdot X(\kappa_{k'[\alpha_s]},\kappa_{j^{(k')}_{l'}[\alpha_s]}):\,.
\end{align*}
Then we have the following lemma.
\begin{lemma}\label{eq:rs}
If the matrices $A[\alpha_r]$ and $A[\alpha_s]$ are non-crossing, then we have
\[
\frac{(\kappa_{k[\alpha_r]}-\kappa_{k'[\alpha_s]})(\kappa_{j^{(k)}_l[\alpha_r]}-\kappa_{j^{(k')}_{l'}[\alpha_s]})}{(\kappa_{k[\alpha_r]}-\kappa_{j^{(k')}_{l'}[\alpha_s]})
(\kappa_{j^{(k)}_l[\alpha_r]}-\kappa_{k'[\alpha_s]})}>0.
\]
\end{lemma}
\begin{proof}
Since the permutation $A[\alpha_{r}]$ and $A[\alpha_{s}]$ are non-crossing, we have either
\begin{align*}
\kappa_{k[\alpha_r]}<\kappa_{k'[\alpha_s]}<\kappa_{j^{(k')}_{l'}[\alpha_s]}<\kappa_{j^{(k)}_{l}[\alpha_r]}\quad\quad\text{or}\quad\quad \kappa_{k[\alpha_r]}<\kappa_{j^{(k)}_{l}[\alpha_r]}<\kappa_{k'[\alpha_s]}<\kappa_{j^{(k')}_{l'}[\alpha_s]}.
\end{align*}
These orderings implies the assertion.
\end{proof}

Then it is immediate to have the following proposition.
\begin{proposition}\label{prop:regularity}
Suppose that $A[\alpha_k]\in\Gr(n_k,m_k)_{\ge 0}$ for $k=1,\ldots,K$ are mutually non-crossing. Let $V[\alpha_k]$ be the vertex operators
associated with these matrices. Then the $\tau$-function \eqref{eq:tau-VO} generated by these vertex operators
gives a KP soliton (regular) associated with a matrix in $\Gr(N,M)_{\ge0}$, where $N=n_1+\cdots+n_K$ and $M=m_1+\cdots+m_K$.
\end{proposition}

One should remark that the $\tau$-function \eqref{eq:tau-VO} is different from that associated with the totally nonnegative matrix given by the $\kappa$-direct sum
of these matrices. However, those solutions have the same permutation, that is, their corresponding matrices are different but they are in the same cell in $\Gr(N,M)_{\ge 0}$.
We give the following example to illustrate the remark.

\begin{example}
We continue Example \ref{ex:17} ($\ell=7$) with
\[
\pi(A[\alpha_1])=(1,6,9,5),\qquad \pi(A[\alpha_2])=(2,3,4)(7,8).
\]
where the numbers are in the sorted coordinates.
The corresponding totally nonnegative matrices are
\begin{align*}
A[\alpha_{1}]=\left(
\begin{array}{cccc}
1 & a_{1,5} & 0 & a_{1,9}\\
0 & 0 &  1 & a_{4,9}\\
\end{array}
\right),\qquad
A[\alpha_2]=\begin{pmatrix}
1 & 0 &a_{2,4} & 0 & 0 \\
0 & 1 & a_{3,4} & 0 & 0\\
0 & 0 & 0 & 1 & a_{5,8}
\end{pmatrix},
\end{align*}
where the signs of the nonzero elements $a_{i,j}$ can be determined by Lemma \ref{eq:sign} (see below).

From the nonzero elements in the matrices $A[\alpha_1]$ and $A[\alpha_2]$, the vertex operators for these matrices are given by
\begin{align*}
V[\alpha_1]&=\exp\left(b_{1,5}X(\kappa_1,\kappa_5)+b_{1,9}X(\kappa_1,\kappa_9)+
b_{4,9}X(\kappa_6,\kappa_9)\right),\\
V[\alpha_2]&=\exp\left(b_{2,4}X(\kappa_2,\kappa_4)+b_{3,4}X(\kappa_3,\kappa_4)+
b_{5,8}X(\kappa_7,\kappa_8)\right).
\end{align*}
Here the positive matrices $B[\alpha_1]$ and $B[\alpha_2]$ are given by
\[
B[\alpha_1]=\begin{pmatrix} b_{1,5}&b_{1,9}\\0& b_{4,9}\end{pmatrix},\qquad
B[\alpha_2]=\begin{pmatrix} b_{2,4}&0\\ b_{3,4}&0\\0&b_{5,8}\end{pmatrix},
\]
where
\begin{align*}
b_{1,5}=a_{1,5}\frac{\kappa_6-\kappa_5}{\kappa_6-\kappa_1},\quad
b_{1,9}=a_{1,9}\frac{\kappa_6-\kappa_9}{\kappa_6-\kappa_1},\quad
b_{4,9}=a_{4,9}\frac{\kappa_1-\kappa_9}{\kappa_1-\kappa_6},
\end{align*}
and
\begin{align*}
b_{2,4}&=a_{2,4}\frac{(\kappa_3-\kappa_4)(\kappa_7-\kappa_4)}{(\kappa_3-\kappa_2)(\kappa_7-\kappa_2)},
\quad
b_{3,4}=a_{3,4}\frac{(\kappa_2-\kappa_4)(\kappa_7-\kappa_4)}{(\kappa_2-\kappa_3)(\kappa_7-\kappa_3)},\\
b_{5,8}&=a_{5,8}\frac{(\kappa_2-\kappa_8)(\kappa_3-\kappa_8)}{(\kappa_2-\kappa_7)(\kappa_3-\kappa_7)}.
\end{align*}
Note here that Lemma \ref{eq:sign} shows that the sign $a_{1,9}, a_{2,4}<0$ and others are positive.

From these vertex operators, we have the $\tau$-functions given by $\tau(\t,\alpha_j)=V[\alpha_j]\cdot 1$, i.e.
\begin{align*}
\tau(\t,\alpha_1) =&
1+b_{1,5}e^{\phi_{1,5}}+b_{1,9}e^{\phi_{1,9}}+b_{4,9}e^{\phi_{6,9}}+b_{1,5}b_{4,9}\frac{(\kappa_1-\kappa_5)(\kappa_6-\kappa_9)}{(\kappa_1-\kappa_9)(\kappa_5-\kappa_6)}e^{\phi_{1,5}+\phi_{6,9}},\\
\tau(\t,\alpha_2)=&1+b_{2,4}e^{\phi_{2,4}}+b_{3,4}e^{\phi_{3,4}}+b_{5,8}e^{\phi_{7,8}}+b_{2,4}b_{5,8}\frac{(\kappa_2-\kappa_7)(\kappa_4-\kappa_8)}{(\kappa_2-\kappa_8)(\kappa_4-\kappa_7)}e^{\phi_{2,4}+\phi_{7,8}}+\\
&+ b_{3,4}b_{5,8}\frac{(\kappa_3-\kappa_7)(\kappa_4-\kappa_8)}{(\kappa_3-\kappa_8)(\kappa_4-\kappa_7)}e^{\phi_{3,4}+\phi_{7,8}},
\nonumber
\end{align*}
where $\phi_{i,j}=\xi_j(\t)-\xi_i(\t)$.

The $\tau$-function generated by those vertex operators are given by
\[
\tau(\t;\alpha_1,\alpha_2)=V[\alpha_1]\cdot V[\alpha_2]\cdot 1.
\]
Since all the operators $X(\kappa_i,\kappa_j)$ commute, we have
\begin{align*}
V[\alpha_1]\cdot V[\alpha_2]=\exp\left(b_{1,5}X(\kappa_1,\kappa_5)+b_{1,9}X(\kappa_1,\kappa_9)+
b_{4,9}X(\kappa_6,\kappa_9)\right.\label{eq:V12}\\
\left. +b_{2,4}X(\kappa_2,\kappa_4)+b_{3,4}X(\kappa_3,\kappa_4)+
b_{5,8}X(\kappa_7,\kappa_8)\right).\nonumber
\end{align*}

Now we give the vertex operator associated with
the totally nonnegative matrix $A[\alpha_1,\alpha_2]$ from the $\kappa$-direct sum $A[\alpha_1]\hat\oplus A[\alpha_2]$, which is given by
\begin{align*}
A[\alpha_1,\alpha_2]=\left(
\begin{array}{ccccccccc}
1 & 0 & 0 & 0 & a_{1,5} & 0 & 0 & 0 & a_{1,9}\\
0 & 1 & 0 & a_{2,4} & 0 & 0 & 0 & 0 &0\\
0 & 0 & 1 & a_{3,4} & 0 & 0 & 0 & 0 &0\\
0 & 0 & 0 & 0 & 0 & 1 & 0 & 0 & a_{4,9}\\
0 & 0 & 0 & 0 & 0 & 0 & 1 & a_{5,8} &0\\
\end{array}
\right).
\end{align*}
Then the corresponding positive matrix $B[\alpha_1,\alpha_2]$ is
\[
B[\alpha_1,\alpha_2]=\begin{pmatrix} 0&c_{1,5}&0&c_{1,9}\\c_{2,4}& 0& 0 & 0\\ c_{3,4} & 0& 0& 0\\
0& 0& 0&c_{4,9}\\0&0&c_{5,8}&0\end{pmatrix},
\]
where the positive elements $c_{i,j}$ are given by, for example,
\begin{align*}
c_{1,5}=a_{1,5}\frac{(\kappa_2-\kappa_5)(\kappa_3-\kappa_5)(\kappa_6-\kappa_5)(\kappa_7-\kappa_5)}{(\kappa_2-\kappa_1)(\kappa_3-\kappa_1)(\kappa_6-\kappa_1)(\kappa_7-\kappa_1)},\\
c_{2,4}=a_{2,4}\frac{(\kappa_1-\kappa_4)(\kappa_3-\kappa_4)(\kappa_6-\kappa_4)(\kappa_7-\kappa_4)}{(\kappa_1-\kappa_2)(\kappa_3-\kappa_2)(\kappa_6-\kappa_2)(\kappa_7-\kappa_2)}.
\end{align*}
Note that we have $a_{2,4}, a_{4,9}<0$ and others are positive (i.e. they have the different signs in the cases $V[\alpha_1]$ and $V[\alpha_2]$).

Then the $\tau$-function associated with $A[\alpha_1,\alpha_2]$ is given by
\begin{equation*}\label{eq:tau12}
\tau_{A[\alpha_1,\alpha_2]}(\t)=V[\alpha_1,\alpha_2]\cdot 1,
\end{equation*}
where the vertex operator $V[\alpha_1,\alpha_2]$ associated with $A[\alpha_1,\alpha_2]$ is given by
\begin{align*}
V[\alpha_1,\alpha_2]=&\exp\left(c_{1,5}X(\kappa_1,\kappa_5)+c_{1,9}X(\kappa_1,\kappa_9)+c_{4,9}X(\kappa_6,\kappa_9)+\right.\\
&\left.+c_{2,4}X(\kappa_2,\kappa_4)+c_{3,4}X(\kappa_3,\kappa_4)+c_{5,8}X(\kappa_7,\kappa_8)\right).
\end{align*}
One should note that this vertex operator is not the same as $V[\alpha_1]\cdot V[\alpha_2]$.
We have the following relations, e.g.
\begin{align*}
&c_{1,5}=b_{1,5}\frac{(\kappa_2-\kappa_5)(\kappa_3-\kappa_5)(\kappa_7-\kappa_5)}{(\kappa_2-\kappa_1)(\kappa_3-\kappa_1)(\kappa_7-\kappa_1)},\qquad
c_{2,4}=b_{2,4}\frac{(\kappa_1-\kappa_4)(\kappa_6-\kappa_4)}{(\kappa_1-\kappa_2)(\kappa_6-\kappa_2)}.
\end{align*}
Notice that both $b_{i,j}$ and $c_{i,j}$ are positive. {The solutions generated by
$V[\alpha_1]\cdot V[\alpha_2]$ and $V[\alpha_1,\alpha_2]$ are both regular and have the same set of solitons. However, the phases of the solitons are different due to the different positive matrices in the vertex operators.}
We might say that they are topologically the same, and we may say $V[\alpha_1,\alpha_2]\equiv V[\alpha_1]\cdot V[\alpha_2]$.

\end{example}

\bigskip
\noindent
{\bf Acknowledgements.}
We would like to thank the referees for valuable comments and informing us the references \cite{W:16} and \cite{Z:02}. This article is dedicated to Professor Ke Wu at Capital Normal University in celebration of his 80th birthday.
One of the authors (C.L.) is supported by National Natural Science Foundation of China (Grant No. 12071237).

 \bigskip
 \noindent
 {\bf Declarations.}
 \begin{itemize}
\item Data sharing not applicable to this article as no datasets were generated or analyzed during the current study.\\
 \item The authors declare no conflicts of interest associated with this manuscript.
 \end{itemize}


\raggedright


\end{document}